\documentclass[a4paper]{amsart}\usepackage[a4paper,margin=2.5cm,includefoot]{geometry}
\usepackage[utf8]{inputenc}
\usepackage[english]{babel}
\usepackage[T1]{fontenc}

\usepackage{amsmath,amssymb,mathtools,bm}
\usepackage{csquotes}
\usepackage{amstext}
\usepackage{amsthm}
\usepackage{bbm}
\usepackage{enumerate}
\usepackage{hyperref}
\usepackage[nameinlink,capitalise]{cleveref}
\usepackage{braket}
\usepackage{dsfont}
\usepackage{color}
\usepackage{tikz}
\usepackage{pgfplots}

\makeatletter
\def\@seccntformat#1{%
	\protect\textup{\protect\@secnumfont
		\ifnum\pdfstrcmp{subsection}{#1}=0 \bfseries\fi
		\ifnum\pdfstrcmp{subsubsection}{#1}=0 \itshape\fi
		\csname the#1\endcsname
		\protect\@secnumpunct
	}%
}
\renewcommand{\@upn}{}
\DeclareRobustCommand{\crefnosort}[1]{%
	\begingroup\@cref@sortfalse\cref{#1}\endgroup
}
\makeatother

\numberwithin{equation}{section}
\newtheorem{thm}{Theorem}[section]
\newtheorem{lem}[thm]{Lemma}
\AddToHook{env/lem/begin}{\crefalias{thm}{lem}}
\newtheorem{prop}[thm]{Proposition}
\AddToHook{env/prop/begin}{\crefalias{thm}{prop}}
\newtheorem{cor}[thm]{Corollary}
\AddToHook{env/cor/begin}{\crefalias{thm}{cor}}

\theoremstyle{definition}
\newtheorem{defn}[thm]{Definition}
\AddToHook{env/defn/begin}{\crefalias{thm}{defn}}

\AddToHook{env/hyp/begin}{\crefalias{thm}{hyp}}
\renewcommand*{\thehyp}{\Alph{hyp}}

\theoremstyle{remark}
\newtheorem{rem}[thm]{Remark}
\AddToHook{env/rem/begin}{\crefalias{thm}{rem}}
\newtheorem{ex}[thm]{Example}
\AddToHook{env/ex/begin}{\crefalias{thm}{ex}}

\crefname{hyp}{Hypothesis}{Hypotheses}\Crefname{hyp}{Hypothesis}{Hypotheses}
\crefname{lem}{Lemma}{Lemmas}\Crefname{lem}{Lemma}{Lemmas}
\crefname{thm}{Theorem}{Theorems}\Crefname{thm}{Theorem}{Theorems}
\crefname{prop}{Proposition}{Propositions}\Crefname{prop}{Proposition}{Propositions}
\crefname{enumi}{}{}\Crefname{enumi}{}{}
\creflabelformat{enumi}{#2(#1)#3}
\crefname{equation}{}{}\Crefname{equation}{}{}
\crefname{rem}{Remark}{Remarks}\Crefname{rem}{Remark}{Remarks}
\crefname{ex}{Example}{Examples}\Crefname{ex}{Example}{Examples}

\makeatletter
\renewcommand{\@upn}{} 
\makeatother

\usepackage[inline]{enumitem}
\newlist{enumthm}{enumerate}{1} 
\setlist[enumthm]{label=\upshape(\roman*),ref=\thethm\,(\roman*)}  
\crefalias{enumthmi}{thm} 
\newlist{enumcor}{enumerate}{1}
\setlist[enumcor]{label=\upshape(\roman*),ref=\thecor\,(\roman*)}
\crefalias{enumcori}{cor}
\newlist{enumlem}{enumerate}{1}
\setlist[enumlem]{label=\upshape(\roman*),ref=\thelem\,(\roman*)}
\crefalias{enumlemi}{lem}
\newlist{enumprop}{enumerate}{1}
\setlist[enumprop]{label=\upshape(\roman*),ref=\theprop\,(\roman*)}
\crefalias{enumpropi}{prop}
\newlist{enumhyp}{enumerate}{1}
\setlist[enumhyp]{label=\upshape(\roman*),ref=\thehyp\,(\roman*)}
\crefalias{enumhypi}{hyp}
\newlist{enumproof}{enumerate*}{1}
\setlist[enumproof]{label=\upshape(\roman*)}
\newlist{enumdef}{enumerate}{1}
\setlist[enumdef]{label=\upshape(\roman*),ref=\thedefn\,(\roman*)}
\crefalias{enumdefi}{defn}

\makeatletter
\newcounter{subcreftmpcnt} %
\newcommand\romansubformat[1]{(\roman{#1})} 
\def\subcref{\@ifstar\@@subcref\@subcref}
\newcommand\@subcref[2][\romansubformat]{%
	\ifcsname r@#2@cref\endcsname
	\cref@getcounter {#2}{\mylabel}%
	\setcounter{subcreftmpcnt}{\mylabel}%
	\hyperref[#2]{\romansubformat{subcreftmpcnt}}%
	\else ?? \fi}   
\newcommand\@@subcref[2][\romansubformat]{%
	\ifcsname r@#2@cref\endcsname
	\cref@getcounter {#2}{\mylabel}%
	\setcounter{subcreftmpcnt}{\mylabel}%
	\romansubformat{subcreftmpcnt}%
	\else ?? \fi}   
\makeatother

\makeatletter
\DeclareRobustCommand{\crefnosort}[1]{%
	\begingroup\@cref@sortfalse\cref{#1}\endgroup
}
\makeatother

\def\endstepsymbol{$\lozenge$}
\def\endclaimsymbol{$\lozenge$}
\newcounter{proofstep}
\AtBeginEnvironment{proof}{\setcounter{proofstep}{0}}

\crefname{proofstep}{Step}{Steps}
\Crefname{proofstep}{Step}{Steps}
\newcounter{proofclaim}
\AtBeginEnvironment{proof}{\setcounter{proofclaim}{0}}

\crefname{proofclaim}{Claim}{Claims}
\Crefname{proofclaim}{Claim}{Claims}

\newcommand{\cB}{{\mathcal B}}
\newcommand{\cF}{{\mathcal F}}
\newcommand{\cH}{{\mathcal H}}

\newcommand{\cM}{{\mathcal M}}
\newcommand{\cQ}{{\mathcal Q}}

\newcommand{\fF}{{\mathfrak F}}

\newcommand{\fM}{{\mathfrak M}}
\newcommand{\fQ}{{\mathfrak Q}}

\newcommand{\fh}{{\mathfrak h}}

\newcommand{\BC}{{\mathbb C}}

\newcommand{\BN}{{\mathbb N}}
\newcommand{\BR}{{\mathbb R}}

\newcommand{\BZ}{{\mathbb Z}}

\usepackage{dsfont} 

\newcommand{\dsone}{{\mathds 1}}

\usepackage{mathrsfs}

\newcommand{\sD}{{\mathscr D}}

\newcommand{\sM}{{\mathscr M}}

\newcommand{\sfd}{{\mathsf d}}\newcommand{\sff}{{\mathsf f}}

\newcommand{\sfs}{{\mathsf s}}

\newcommand{\rmd}{{\mathrm d}}

\newcommand{\IR}{\BR}\newcommand{\IC}{\BC}
\newcommand{\N}{\BN}\newcommand{\Z}{\BZ}\newcommand{\R}{\BR}\newcommand{\C}{\BC}

\newcommand{\hs}{\fh}\newcommand{\HS}{\cH}

\newcommand{\eps}{\varepsilon}\newcommand{\ph}{\varphi}

\newcommand{\Id}{\dsone} 

\newcommand{\supp}{\operatorname{supp}}

\DeclareFontFamily{U}{mathx}{\hyphenchar\font45}
\DeclareFontShape{U}{mathx}{m}{n}{
	<5> <6> <7> <8> <9> <10>
	<10.95> <12> <14.4> <17.28> <20.74> <24.88>
	mathx10
}{}
\DeclareSymbolFont{mathx}{U}{mathx}{m}{n}
\DeclareFontSubstitution{U}{mathx}{m}{n}
\DeclareMathAccent{\widecheck}{0}{mathx}{"71}
\DeclareMathAccent{\wideparen}{0}{mathx}{"75}

\DeclareFontFamily{OMX}{MnSymbolE}{}
\DeclareFontShape{OMX}{MnSymbolE}{m}{n}{
	<-6>  MnSymbolE5
	<6-7>  MnSymbolE6
	<7-8>  MnSymbolE7
	<8-9>  MnSymbolE8
	<9-10> MnSymbolE9
	<10-12> MnSymbolE10
	<12->   MnSymbolE12}{}
\DeclareSymbolFont{mnlargesymbols}{OMX}{MnSymbolE}{m}{n}
\SetSymbolFont{mnlargesymbols}{bold}{OMX}{MnSymbolE}{b}{n}
\DeclareMathDelimiter{\llangle}{\mathopen}{mnlargesymbols}{'164}{mnlargesymbols}{'164}
\DeclareMathDelimiter{\rrangle}{\mathclose}{mnlargesymbols}{'171}{mnlargesymbols}{'171}
\DeclareMathDelimiter{\lsem}{\mathopen}{mnlargesymbols}{'102}{mnlargesymbols}{'102}
\DeclareMathDelimiter{\rsem}{\mathclose}{mnlargesymbols}{'107}{mnlargesymbols}{'107}
\DeclareMathDelimiter{\langlebar}{\mathopen}{mnlargesymbols}{'152}{mnlargesymbols}{'152}
\DeclareMathDelimiter{\ranglebar}{\mathclose}{mnlargesymbols}{'157}{mnlargesymbols}{'157}
\DeclareMathDelimiter{\lWavy}{\mathopen}{mnlargesymbols}{'137}{mnlargesymbols}{'137}
\DeclareMathDelimiter{\rWavy}{\mathopen}{mnlargesymbols}{'137}{mnlargesymbols}{'137}

\newcommand{\FGamma}{\Gamma}
\newcommand{\FS}{\cF}\newcommand{\dG}{\sfd\FGamma}

\title[Wiener-Type Theorems for Laplace Transform]{Wiener-Type Theorems for the Laplace Transform.\\ With Applications to Ground State Problems}
\author{Benjamin Hinrichs}
\address{B. Hinrichs, Universit\"at Paderborn, Institut f\"ur Mathematik, Institut f\"ur Photonische Quantensysteme, Warburger Str. 100, 33098 Paderborn, Germany}
\email{benjamin.hinrichs@math.upb.de}
\author{Steffen Polzer}
\address{S. Polzer, Université de Genève, Section de mathématiques, Rue du Conseil-Général 7-9, 1205 Genève, Switzerland}
\email{steffen.polzer@unige.ch}

\usepackage{xcolor}\usepackage{cancel}

\newcommand{\1}{\mathds{1}}

\begin{document}
	
	\begin{abstract} 
		\noindent
		We study the behavior of a probability measure near the bottom of its support in terms of time averaged quotients of its Laplace transform. We discuss how our results are connected to both rank-one perturbation theory as well as renewal theory. We further apply our results in order to derive criteria for the existence and non-existence of ground states for a finite dimensional quantum system coupled to a bosonic field.
	\end{abstract}
	
	\maketitle
	
	\vspace*{-1em}
	\section{Introduction}
	For a finite Borel measure $\mu$ on the real line $\IR$,
	Wiener's theorem \cite{Wiener.1933}, sometimes also referred to as Wiener's lemma, provides equality of
	the $\ell^2$-norm of its atoms $\sum_x |\mu(\{x\})|^2$
	and the $L^2$-ergodic average $\lim_{T \to \infty}\frac1{2T}\int_{-T}^T|\hat\mu(t)|^2\,\rmd t$ of its Fourier transform $\hat\mu$, see for example \cite[\S\,VI,~Thm.~2.12]{Katznelson.2004}.
	It has many applications in ergodic theory and
	is the main ingredient in the proof of the famous RAGE theorem, see for example \cite[\S\,XI.17]{ReedSimon.1979},
	a key statement in spectral and scattering theory. The latter determines, for a given selfadjoint operator $H$, e.g.\ the Hamiltonian of a quantum system,
 	the large time asymptotics of the solutions $t\mapsto \psi_t=e^{-i tH}\psi_0$ to Schr\"odinger's equation, and thus the dynamics of the quantum system, in terms of the spectral parts of $H$.
	However, the study of the time-dependence of $\psi_t$ usually does not provide explicit information on certain parts of the spectrum $\sigma(H)$,
	since it is somewhat hidden in the fluctuations.
	Thus, especially when interested in studying the low-energy regime close to $E=\inf\sigma(H)$, e.g., the question whether $E$ is an eigenvalue of $H$,
	it is useful to study the semigroup $(e^{-tH})_{t\geq 0}$ and the solutions $t\mapsto e^{-tH}\psi_0$ to the heat equation instead,
	since spectral parts above $E$ will be exponentially suppressed therein for large $t$.
	This is an especially appealing approach,
	since path integral representations of the semigroup provided by Feynman--Kac formulas allow to apply probabilistic techniques,
	see \cite{LorincziHiroshimaBetz.2011,DemuthvanCasteren.2000} for textbook treatises on the subject.
	
	Reformulated in terms of the spectral measure $\mu$ of $H$ taken with respect to a suitable test vector $\phi$,
	the question whether $E$ is an eigenvalue of $H$ is equivalent to asking if $\mu$ is has an atom in $E$. Suitable here means that one needs to ensure that $\phi$ would be non-orthogonal to a potentially existing ground state.
	If $t\mapsto Z_t = \langle \phi, e^{-tH} \phi \rangle$ denotes the Laplace transform of $\mu$, one then needs to check whether $\mu(\{E\}) = \lim_{t\to \infty}e^{Et} Z_t$ is positive or vanishes.
	However, doing this directly would require in particular a very good understanding of the precise value of $E$,
	which in general can not be expected.
	It has previously been noted that one one might circumvent this problem by studying the limit of the quotient $Z_t^2/Z_{2t}$ as $t\to \infty$ instead,
	where a (non-)zero limit of this quotient implies that $\mu(\{E\})$ is (non-)zero, see again \cite{LorincziHiroshimaBetz.2011}.
	A more general treatment of this approach is, however, unknown to the authors.
	
	We here fill this gap and study a probability measure near the bottom of its support in terms of time-averaged quotients of its Laplace transform.
	More precisely, we especially prove a novel formula expressing $\mu(\{E\})$ as an ergodic average over quotients of the form $Z_sZ_{t-s}/Z_t$.
	Furthermore, we express the moment $\int_{(E,\infty)} \frac{1}{x-E}\,\mu(\rmd x)$ in terms of such ergodic averages, at least under the additional assumption $\mu(\{E\})>0$. All these results are collected and proven in \cref{Section: Wiener Type Theorems}.
	
	In view of our above motivation we further provide three applications of our results: (1) We connect them to rank-one perturbation theory of selfadjoint operators, yielding a natural interpretation of our results from a functional analytic point of view (\cref{Section: Perturbation Theory}). (2) We relate them to renewal theory, thus providing a natural interpretation in terms of probabilistic notions (\cref{Section: Renewal theory}). (3) We extend known results on ground state existence and absence for so-called generalized spin boson models as an important application of our formulas (\cref{Section: Generalized Spin-Boson models}). 
	
	\subsection*{Acknowledgments} BH acknowledges support by the Ministry of Culture and Science of the
	State of North Rhine-Westphalia within the project `PhoQC' (Grant
	Nr. PROFILNRW-2020-067). SP acknowledges funding from the Swiss State
	Secretariat for Education, Research and Innovation (SERI) through the consolidator grant ProbQuant,
	and funding from the Swiss National Science Foundation through the NCCR SwissMAP grant.
	
	\section{Wiener-Type Theorems}
	\label{Section: Wiener Type Theorems}
	In the following, let $\mu$ be a Borel probability measure on $\R$ whose support is bounded from below,
	where we as usual define the support of a Borel measure as the set of all points of which every open neighborhood has positive measure.
	Now let $E \coloneqq \inf \operatorname{supp}(\mu)$
	and let $Z = (Z_t)_{t\geq 0}$ be the Laplace transform of $\mu$, i.e.,
	\begin{align}
		\label{def:Laplace}
		Z_t \coloneqq \int_{[E, \infty)} e^{-tx} \, \mu(\mathrm dx)\quad \text{ for } t \geq 0.
	\end{align}
	It is well-known that the value of $E$ can be studied using $Z_t$, e.g., by employing the formula
	\begin{align}
		\label{eq:infsupp}
		E = -\lim_{t \to \infty}\frac1t\log Z_t,
	\end{align}
	which in turn follows from the simple estimate $e^{-t(E+\eps)}\mu([E,E+\eps))\le Z_t \le e^{-tE}$ for arbitrary $\eps>0$.
	Thus, there is an inherent connection between the exponential behavior of the Laplace transform
	and the infimum of the support of $\mu$.
	
	Exploiting this exponential behavior further allows to study the size of a (possible) atom of $\mu$ at $E$ in terms of $Z_t$,
	which is the result of the following Wiener-type formula.
	Notably, the ratios studied therein allow us to characterize $\mu(\{E\})$ without any knowledge of the exact value of $E$,
	which makes it especially useful in settings where the Laplace transform of $\mu$ is tractable but the calculation of $E$ remains complicated.
	\begin{thm}
		\label{Theorem: Convergence of averaged partition function}
		For any $\kappa \in (0, 1)$
		\begin{equation}
			\label{Equation: atom in E as a limit}
			\mu(\{E\}) = \lim_{t\to \infty} \frac{Z_{\kappa t} Z_{(1-\kappa)t}}{Z_t}  = \lim_{t\to \infty} \frac{1}{t}\int_0^t\frac{Z_{s} Z_{t-s}}{Z_t} \, \mathrm ds.
		\end{equation}
	\end{thm}
	\begin{rem}
		At least in the context of field-matter interactions described by Feynman--Kac formulas,
		as we will discuss in more detail in \cref{Section: Application to Spin systems},
		the case $\kappa=\frac12$ of the first identity is well-known
		and has been applied in various articles, see for example \cite{LorincziHiroshimaBetz.2011}
		for a textbook version.
		However, neither the case $\kappa\ne\frac12$ nor the novel averaging formula on the right hand side
		have to the authors knowledge appeared in the literature before.
	\end{rem}
	
	While the proof of \cref{Theorem: Convergence of averaged partition function} is elementary,
	we will see in the next section that it has a natural interpretation in terms of rank one perturbation theory.
	This connection will also further motivate the following result, which can again be shown by elementary means
	and is hence presented here as well.
	\begin{thm}
		\label{Theorem: Second order term}
		Assume that $\mu({\{E\}}) > 0$ and that
		\begin{equation}
			\label{Equation: decay close to bottom of spectrum}
			\int_{(E, \infty)}\frac{\mu(\mathrm dx)}{x-E} < \infty.
		\end{equation}
		Then 
		\begin{equation}
			\label{Equation: Second order term}
			\int_{(E, \infty)} \cfrac{\mu(\mathrm dx)}{x-E} =\, \, \lim_{t\to \infty} \, \, \cfrac{2\bigg(\displaystyle\int_0^t \mathrm ds \int_0^s \mathrm dr \, \cfrac{Z_{t-s} Z_{s-r} Z_r}{Z_t}\bigg) - \bigg(\int_0^t \, \mathrm ds\cfrac{Z_{t-s} Z_s}{Z_t}  \bigg)^2}{2\displaystyle\int_0^t \, \mathrm ds\cfrac{Z_{t-s} Z_s}{Z_t}}.
		\end{equation}
	\end{thm}
	In particular, if the right hand side of \cref{Equation: Second order term} is infinite then so is the left hand side. While it might be desirable to show that finiteness of the right hand side also implies finitness of the left hand side, i.e.,\ that \cref{Theorem: Second order term} remains true even without Assumption \cref{Equation: decay close to bottom of spectrum}, we will leave this for further research. It should be noted, however, that in \cref{Theorem: Renewal representation Theorem 1 and Theorem 2} below we will give a condition in terms of renewal theory that is both sufficient as well as necessary for \cref{Equation: decay close to bottom of spectrum} to hold.
	
	The remainder of this section is devoted to the proof of these main results and can be skipped by readers more interested in their applications in the subsequent sections.
	
	\subsection{Proof of \cref{Theorem: Convergence of averaged partition function}}
	Before proving our first result, let us note the following useful properties of the ratios of Laplace transforms.
	\begin{prop}
		\label{Proposition: Monotonicty of fraction of partition functions}
		The following holds:
		\begin{enumerate}
			\item For any $t\geq 0$ the function
			\begin{equation}
				\label{Equation: quotient of partition functions}
				[0, t] \to (0, \infty), \quad s \mapsto \frac{Z_s Z_{t-s}}{Z_t}
			\end{equation}
			is decreasing on $[0, t/2]$ and increasing on $[t/2, t]$. In particular,
			\begin{equation*}
				\max_{0\leq s \leq t} \frac{Z_s Z_{t-s}}{Z_t} = 1, \quad \min_{0\leq s \leq t}\frac{Z_s Z_{t-s}}{Z_t} = \frac{Z_{t/2}^2}{Z_t}.
			\end{equation*} 
			\item For any $s \geq 0$ the function
			\begin{equation}
				\label{Equation: quotient of partition functions 2}
				[s, \infty) \to (0, \infty), \quad t \mapsto \frac{Z_s Z_{t-s}}{Z_t}
			\end{equation}
			is decreasing.
			\item The function
			\begin{equation}
				\label{Equation: Average quotient of partition functions}
				(0, \infty) \to \R, \quad t \mapsto \frac{1}{t} \int_0^t \frac{Z_s Z_{t-s}}{Z_t} \, \mathrm ds 
			\end{equation}
			is decreasing.
		\end{enumerate}
	\end{prop}
	\begin{proof}
		After a translation of $\mu$, i.e., eventually replacing $\mu$ by $\mu_E(\,\cdot\,)=\mu(\,\cdot\, +E)$ and observing  that this changes the Laplace transform to $e^{tE}Z_t$ thus leaving the ratio in \cref{Equation: quotient of partition functions} invariant,
		we might assume that $E = 0$.
		Further, we note that we might differentiate for $t>0$ under the integral (by the dominated convergence theorem) such that
		\[
		\frac{\rmd }{\rmd t}Z_t = \int_{[0,\infty)}(-x)e^{-tx} \mu(\rmd x)
		\]
		We may thus calculate the derivative in $s\in (0, t)$ as
		\begin{align*}
			&  {Z_s^{-2}} \cdot \frac{\mathrm d}{\mathrm ds} (Z_s Z_{t-s})\\
			&= Z_s^{-1}\int_{[0, \infty)} x e^{-(t-s)x} \, \mu(\mathrm dx)   - Z_{s}^{-2} \Big(\int_{[0, \infty)} xe^{-sx} \, \mu(\mathrm dx) \Big) \Big(\int_{[0, \infty)} e^{-(t-s)x} \, \mu(\mathrm dx) \Big)\\
			&= \int_{[0, \infty)} x e^{-(t-2s)x} \, \widehat{\mu}_s(\mathrm dx)  - \Big(\int_{[0, \infty)} x  \, \widehat{\mu}_s(\mathrm dx) \Big) \Big( \int_{[0, \infty)} e^{-(t-2s)x} \, \widehat{\mu}_s(\mathrm dx) \Big)
		\end{align*}
		where
		\begin{equation*}
			\widehat{\mu}_s(\mathrm dx) \coloneqq Z_s^{-1} e^{-sx} \, \mu(\mathrm dx).
		\end{equation*}
		We can now apply the FKG inequality
		\begin{align}
			\label{eq:FKG}
			\int fg\, \rmd \nu \ge \int f\, \rmd \nu \int g\,\rmd \nu
		\end{align}
		for probability measures $\nu$ on $[0, \infty)$ and $f$ and $g$ both having the same type of monotonicity, see for example \cite[\textsection~2.2]{Grimmett.1999} for a proof.
		Note that the inequality reverses, if one function is increasing and the other is decreasing.
		Applying \cref{eq:FKG} with $\nu=\widehat \mu_s$ thus yields the monotonicity of \cref{Equation: quotient of partition functions} on $[0, t/2]$ and $[t/2, t]$.
		
		In the same manner,
		\begin{align*}
			\frac{\mathrm d}{\mathrm dt}\frac{Z_{t-s}}{Z_t} = \Big(\int_{[0, \infty)} x  \, \widehat{\mu}_t(\mathrm dx) \Big) \Big( \int_{[0, \infty)} e^{sx} \, \widehat{\mu}_t(\mathrm dx) \Big) -\int_0^\infty x e^{sx} \, \widehat \mu_t(\mathrm ds) \leq 0
		\end{align*}
		by \cref{eq:FKG}, which shows that \cref{Equation: quotient of partition functions 2} is decreasing.
		
		It is left to show that \cref{Equation: Average quotient of partition functions} is decreasing.
		This however follows, since by the previous considerations for any $t>0$ and $\alpha>1$
		\begin{align*}
			\frac{1}{\alpha t} \int_0^{\alpha t} \frac{Z_s Z_{\alpha t-s}}{Z_{\alpha t}} \, \mathrm ds &= \frac{1}{t} \int_0^{t} \frac{Z_{\alpha s} Z_{\alpha t-\alpha s}}{Z_{\alpha t}}\, \mathrm ds 
			= \frac{2}{t} \int_0^{t/2} \frac{Z_{\alpha s} Z_{\alpha t-\alpha s}}{Z_{\alpha t}}\, \mathrm ds \\
			&\leq \frac{2}{t} \int_0^{t/2} \frac{Z_{s} Z_{\alpha t- s}}{Z_{\alpha t}}\, \mathrm ds 
			\leq \frac{2}{t} \int_0^{t/2} \frac{Z_{s} Z_{t- s}}{Z_{t}} = \frac{1}{t} \int_0^{t} \frac{Z_{s} Z_{t- s}}{Z_{t}} \, \mathrm ds. \qedhere
		\end{align*}
	\end{proof}
	We move to the 
	\begin{proof}[Proof of \cref{Theorem: Convergence of averaged partition function}]
		With out loss of generality we again assume that $E = 0$.
		Let us first consider the case $\mu(\{0\}) > 0$. Then by the dominated convergence theorem
		\begin{equation*}
			\lim_{t\to \infty} Z_t =  \mu(\{0\}).
		\end{equation*}
		Hence, for any $\kappa \in (0, 1)$,
		\begin{equation*}
			\lim_{t\to \infty} \frac{Z_{\kappa t} Z_{(1-\kappa)t}}{Z_t} = \mu(\{0\}).
		\end{equation*}
		Moreover, we have by Fubinis Theorem
		\begin{align}
			\label{eq:Fubiniaverage}
			\begin{aligned}
				\frac{1}{t}\int_0^t Z_s Z_{t-s} \, \mathrm ds &= \int_{[0, \infty)} \mu(\mathrm dx) \int_{[0, \infty)} \mu(\mathrm dy) \, e^{-ty}\frac{1}{t}\int_0^t \mathrm ds \, e^{s(y-x)}
				\\&= \int_{[0, \infty)} \mu(\mathrm dx) \int_{[0, \infty)} \mu(\mathrm dy) f_t(x, y)
			\end{aligned}
		\end{align}
		where
		\begin{equation}
			\label{eq:ft}
			f_t(x, y) \coloneqq \frac{e^{-tx} - e^{-ty}}{t(y-x)} \1_{\{x\neq y\}} + e^{-ty} \1_{\{x = y\}}.
		\end{equation}
		Notice that $f_t(x, y) = f_t(y, x)$ for all $x, y\geq 0$ and, by an application of the mean value theorem,
		\begin{equation}
			\label{eq:ftestimate}
			f_t(x, y) \leq e^{-tx}
		\end{equation}
		for all $t\geq 0$ and $0\leq x\leq y$. By the dominated convergence theorem, we obtain
		\begin{equation*}
			\lim_{t\to \infty} \frac{1}{t}\int_0^t Z_s Z_{t-s} \, \mathrm ds = \int_{[0, \infty)} \mu(\mathrm dx) \int_{[0, \infty)} \mu(\mathrm dy) \1_{\{x = y = 0\}} = \mu(\{0\})^2.
		\end{equation*}
		which concludes the proof of \cref{Equation: atom in E as a limit} for the case that $\mu(\{0\}) > 0$.
		
		Let us now on the contrary assume that $\mu(\{0\}) = 0$.
		First, notice that \cref{eq:Fubiniaverage} implies
		\begin{equation*}
			\frac{1}{t}\int_0^t Z_s Z_{t-s} \, \mathrm ds \leq 2\int_{[0, \infty)} \mu(\mathrm dx) \int_{[x, \infty)} \mu(\mathrm dy) f_t(x, y),
		\end{equation*}
		where equality does not necessarily hold because of potential atoms of $\mu \otimes \mu$ on the diagonal.
		We fix some $\varepsilon > 0$ and split the right hand side
		\begin{equation*}
			\int_{[0, \infty)} \mu(\mathrm dx) \int_{[x, \infty)} \mu(\mathrm dy) f_t(x, y) = T_1(t) + T_2(t) + T_3(t)
		\end{equation*}
		with
		\begin{align*}
			T_1(t) &\coloneqq \int_{[0, \varepsilon]} \mu(\mathrm dx) \int_{[x, x +\varepsilon]} \mu(\mathrm dy) f_t(x, y), \\
			T_2(t) &\coloneqq \int_{[0, \varepsilon]} \mu(\mathrm dx) \int_{(x +\varepsilon, \infty)} \mu(\mathrm dy) f_t(x, y), \\
			T_3(t) &\coloneqq \int_{(\varepsilon, \infty)} \mu(\mathrm dx) \int_{[x, \infty)} \mu(\mathrm dy) f_t(x, y).
		\end{align*}
		Inserting the estimate \cref{eq:ftestimate} in $T_1$ and the definition \cref{eq:ft} in $T_2$, we find
		\begin{align*}
			T_1(t) &\leq \int_{[0, \varepsilon]} \mu(\mathrm dx) \int_{[x, x +\varepsilon]} \mu(\mathrm dy) e^{-tx} \leq \mu([0, 2 \varepsilon]) \cdot Z_t,\\
			T_2(t) &\leq \int_{[0, \varepsilon]} \mu(\mathrm dx) \int_{(x +\varepsilon, \infty)} \mu(\mathrm dy) \frac{e^{-tx}}{t\varepsilon} \leq \frac{1}{t\varepsilon} Z_t.
		\end{align*}
		This leads to
		\begin{equation*}
			\limsup_{t\to \infty} \frac{T_1(t) + T_2(t)}{Z_t} \leq \mu([0, 2 \varepsilon]).
		\end{equation*}
		Moreover, again applying the estimate \cref{eq:ftestimate} and using that  $\mu([0, \varepsilon/2])>0$ as $0$ is the infimum of the support of $\mu$, we find
		\begin{equation*}
			\limsup_{t \to \infty} \frac{T_3(t)}{Z_t} \leq \limsup_{t \to \infty} \frac{\int_{(\varepsilon, \infty)} \mu(\mathrm dx) e^{-tx}}{\mu([0, \varepsilon/2]) e^{-\varepsilon t/2}} = 0.
		\end{equation*}
		Combining the above assumptions, we have shown that
		\begin{equation*}
			\limsup_{t\to \infty} \frac{1}{t}\int_0^t \frac{Z_s Z_{t-s}}{Z_t} \, \mathrm ds \leq 2\mu([0, 2 \varepsilon])
		\end{equation*}
		for all $\varepsilon>0$. Hence, if $\mu(\{0\}) = 0$ then
		\begin{equation*}
			\lim_{t\to \infty} \frac{1}{t}\int_0^t \frac{Z_s Z_{t-s}}{Z_t} \, \mathrm ds = 0.
		\end{equation*}
		Furthermore, for $\kappa\in (0, 1/2]$, \cref{Proposition: Monotonicty of fraction of partition functions} implies
		\begin{equation*}
			\frac{1}{t}\int_0^t \frac{Z_s Z_{t-s}}{Z_t} \, \mathrm ds \geq \frac{1}{t}\int_0^{\kappa t} \frac{Z_s Z_{t-s}}{Z_t} \, \mathrm ds \geq \kappa \frac{Z_{\kappa t} Z_{(1-\kappa) t}}{Z_t}
		\end{equation*}
		and hence 
		\begin{equation*}
			\lim_{t\to \infty} \frac{Z_{\kappa t} Z_{(1-\kappa) t}}{Z_t} = 0. \qedhere
		\end{equation*}
	\end{proof}
	
	\subsection{Proof of \cref{Theorem: Second order term}}
	%
	%
	
	In the spirit of the proof of \cref{Theorem: Convergence of averaged partition function}, the following observation is important in proving \cref{Theorem: Second order term}.
	
	\begin{prop}
		\label{Lemma: Asymptotics of convolution}
		Assume that $E=0$ and that \cref{Equation: decay close to bottom of spectrum} holds. Then there exists a continuous function $R:[0, \infty) \to (0, \infty)$ with $\lim_{t\to \infty} R(t) = 0$ such that for all $t>0$
		\begin{equation*}
			\int_0^t Z_{t-s} Z_s \mathrm ds = t \mu(\{0\}) Z_t + 2 \mu(\{0\}) \int_{(0, \infty)} \frac{\mu(\mathrm dx) }{x} + R(t). 
		\end{equation*}
	\end{prop}
	
	\begin{proof}
		Recalling \cref{eq:Fubiniaverage,eq:ft,eq:ftestimate} from the proof of \cref{Theorem: Convergence of averaged partition function}, by Fubinis theorem we have
		\begin{equation}
			\label{Equation: Convolution of partition functions as integral over mu}
			\int_0^t Z_{t-s} Z_s \, \mathrm ds = \int_{[0, \infty)^2} \mu^{\otimes 2}(\mathrm dx \mathrm dy)\, g_t(x, y)
		\end{equation}
		where the function $g_t:[0, \infty)^2 \to [0, \infty)$ is defined by
		\begin{equation*}
			g_t(x, y) \coloneqq \frac{e^{-tx} - e^{-ty}}{(y-x)} \1_{\{x\neq y\}} + te^{-ty} \1_{\{x = y\}}
		\end{equation*}
		which satisfies
		\begin{equation*}
			g_t(x, y) \leq t e^{- t \operatorname{min}(x, y)}
		\end{equation*}
		for all $x, y\in [0, \infty)$. We have
		\begin{equation*}
			T_1(t) \coloneqq \int_{[0, \infty)^2} \mu^{\otimes 2}(\mathrm dx \mathrm dy) \1_{\{x = y = 0\}} g_t(x, 	y) = t \mu(\{0\}) Z_t + R_1(t)
		\end{equation*}
		where
		\begin{equation*}
			R_1(t) \coloneqq t\mu(\{0\})^2 -  t \mu(\{0\}) Z_t = - \mu(\{0\})\int_{(0, \infty)} t e^{-tx} \, \mu(\mathrm dx).
		\end{equation*}
		By \cref{Equation: decay close to bottom of spectrum}, $E=0$ and the dominated convergence theorem with majorant $x\mapsto x^{-1}$, we have
		\begin{align}
			\label{Lemma: Convergence of integral to 0}
			\lim_{t \to \infty} \int_{(0, \infty)} t e^{-tx} \mu(\mathrm dx) = 0,
		\end{align}
		so $R_1(t)\to 0$ as $t\to \infty$. Moreover, utilizing \cref{Lemma: Convergence of integral to 0} once more, we find
		\begin{align*}
			T_2(t) &\coloneqq \int_{[0, \infty)^2} \mu^{\otimes 2}(\mathrm dx \mathrm dy) \1_{\{ \operatorname{min}(x, y)>0  \}} g_t(x, 	y) \\
			&\leq 2 \int_{(0, \infty)} \mu(\mathrm dx) \int_{[x, \infty)} \mu(\mathrm dy) g_t(x, y) \\
			&\leq 2\int_{(0, \infty)} \mu(\mathrm dx) te^{-tx} \xrightarrow{t\to\infty} 0.
		\end{align*}
		Finally, we have
		\begin{align*}
			T_3(t) &\coloneqq \int_{[0, \infty)^2} \mu^{\otimes 2}(\mathrm dx \mathrm dy) \1_{\{ \operatorname{min}(x, y)=0, x\neq y  \}} g_t(x, 	y) \\
			&= 2 \mu(\{0\}) \int_{(0, \infty)} \mu(\mathrm dy) g_t(0, y) \\
			&= 2 \mu(\{0\}) \int_{(0, \infty)} \frac{\mu(\mathrm dy)}{y} + R_3(t)
		\end{align*}
		where
		\begin{equation*}
			R_3(t) \coloneqq -\int_{(0, \infty)} \mu(\mathrm dy) \frac{e^{-ty}}{y} \xrightarrow{t\to\infty} 0
		\end{equation*}
		by  the dominated convergence theorem.
		By \cref{Equation: Convolution of partition functions as integral over mu}, we thus have
		\begin{equation*}
			\int_0^t Z_{t-s} Z_s \, \mathrm ds = T_1(t) + T_2(t) + T_3(t) = t \mu(\{0\}) Z_t + 2 \mu(\{0\}) \int_{(0, \infty)} \frac{\mu(\mathrm dy)}{y} +  R(t)
		\end{equation*}  
		where $R(t) \coloneqq R_1(t) + T_2(t) + R_3(t) \to 0$ as $t\to \infty$. 
	\end{proof}
	We can now give the
	\begin{proof}[Proof of \cref{Theorem: Second order term}]
		Similar to the proof of \cref{Theorem: Convergence of averaged partition function}, as shifting $\mu$ leaves all considered quotients involving $Z$ invariant,
		we can again assume w.l.o.g. that $E=0$.
		By dominated convergence and \cref{Theorem: Convergence of averaged partition function}, we then have
		\begin{align}
			\label{eq:recallfirstthm}
			\lim_{t\to \infty} Z_t = \mu(\{0\}), \quad \lim_{t\to \infty} \frac{1}{t}\int_0^t\frac{Z_{s} Z_{t-s}}{Z_t} \, \mathrm ds = \mu(\{0\}).
		\end{align}
		It is therefore sufficient to show that
		\begin{align}
			\label{eq:goalofproof}
			\lim_{t\to \infty} \frac{1}{t} \bigg[2 Z_t \int_0^t \mathrm ds \int_0^s \mathrm dr \, Z_{t-s} Z_{s-r} Z_r - \bigg(\int_0^t \, \mathrm ds\ Z_{t-s} Z_s  \bigg)^2\bigg] = 2\mu(\{0\})^3\mathcal I
		\end{align}
		with $\mathcal I \coloneqq  \int_{(0, \infty)} x^{-1} \, \mu(\mathrm dx)$.
		Let $R:[0,\infty)\to[0,\infty)$ be chosen as in \cref{Lemma: Asymptotics of convolution}. Then
		\begin{equation*}
			2Z_t \int_0^t \mathrm ds \, Z_{t-s} \int_0^s  \mathrm dr \, Z_{s-r} Z_r = T_1(t) + T_2(t)
		\end{equation*}
		with
		\begin{align*}
			T_1(t) &\coloneqq 2\mu(\{0\}) Z_t  \int_0^t \mathrm ds \, s Z_{t-s} Z_s, \\
			T_2(t) &\coloneqq 2Z_t \int_0^t \mathrm ds \, Z_{t-s}\Big(2\mu(\{0\}) \mathcal I + R(s) \Big).
		\end{align*}
		Now observing that
		\begin{equation*}
			\int_0^t \mathrm ds \, s Z_{t-s} Z_s = \frac{1}{2}\int_0^t \mathrm ds \, s Z_{t-s} Z_s +  \frac{1}{2}\int_0^t \mathrm ds \,(t - s) Z_{t-s} Z_s = \frac{t}{2}\int_0^t \mathrm ds \, Z_{t-s} Z_s,
		\end{equation*}
		we find
		\begin{align*}
			T_1(t) = t \mu(\{0\}) Z_t \int_0^t \mathrm ds \, Z_{t-s} Z_s.
		\end{align*}
		Thus, once more applying \cref{Lemma: Asymptotics of convolution}, we have 
		\begin{equation*}
			\bigg(\int_0^t \mathrm ds \, Z_{t-s} Z_s  \bigg)^2 = T_1(t) + T_3(t)
		\end{equation*}	
		where
		\begin{equation*}
			T_3(t) \coloneqq \bigg(2 \mu(\{0\}) \mathcal I + R(t) \bigg) \int_0^t \, \mathrm ds\ Z_{t-s} Z_s.
		\end{equation*}
		Summarizing the above observations, we have
		\[
		2 Z_t \int_0^t \mathrm ds \int_0^s \mathrm dr \, Z_{t-s} Z_{s-r} Z_r - \bigg(\int_0^t \, \mathrm ds\ Z_{t-s} Z_s  \bigg)^2 = T_2(t)-T_3(t).
		\]
		From \cref{eq:recallfirstthm}, we see
		\[
		\lim_{t\to \infty} \frac{1}{t} T_3(t) = 2\mu(\{0\})^3 \mathcal I.
		\]
		Further, e.g., by Ces\`aros theorem, we have 
		\begin{equation*}
			\lim_{t \to \infty} \frac{1}{t} \int_0^t Z_s \, \mathrm ds = \lim_{t\to \infty} Z_t = \mu(\{0\}), \quad \lim_{t \to \infty} \frac{1}{t} \int_0^t |R(s)|\, \mathrm ds = \lim_{t\to \infty} |R(t)| = 0
		\end{equation*}
		and hence, since $0 \leq Z_{t-s} \leq 1$ for all $0\leq s \leq t$,
		\begin{equation*}
			\lim_{t\to \infty} \frac{1}{t} T_2(t) = 4\mu(\{0\})^3 \mathcal I.
		\end{equation*}
		Combining these observations proves \cref{eq:goalofproof} and thus the statement.
	\end{proof}

	\section{A Link To Rank-One Perturbation Theory}
	\label{Section: Perturbation Theory}
	
	In this section, we connect our theorems to perturbation theory by writing $\mu$ as the spectral measure of a suitable self-adjoint operator.
	By introducing a family of rank one perturbations of said operator, we will see that \cref{Theorem: Convergence of averaged partition function,Theorem: Second order term} correspond to first and second order perturbation theory, respectively.
	Further exploring this connection should allow one to derive higher order analogues of these results by similar means.
	In \cref{Theorem: rank-one perturbation}, we then apply our previous results in order to deduce some fundamental properties of the ground state energy of rank one perturbations of a self-adjoint operator.
	
	We now assume that $H$ is a lower-bounded selfadjoint operator on some Hilbert space $\mathcal H$ and that $\psi\in\HS$ is a unit vector such that
	\begin{equation}
		\label{Equation: Partition function generated by Hamiltonian}
		Z_t = \langle \psi, e^{-tH} \psi \rangle
	\end{equation}
	for all $t\geq 0$, i.e., such that $\mu$ is the spectral measure of $H$ with respect to $\psi$. Moreover, we assume that
	\begin{equation}
		\label{eq:infspec}
		\inf \sigma(H) = \inf \operatorname{supp}(\mu) = E
	\end{equation}
	holds. Notice that all probability measures $\mu$ on $\R$ with lower bounded support have a representation of that form as we might set $\mathcal H = L^2(\R, \mu)$, $(H\phi)(x) = x \phi(x)$ for $\phi$ such that $\int_\R x^2 |\phi(x)|^2 \, \mu(\mathrm dx) < \infty$ and $\psi = 1$.
	\begin{rem}
		\label{rem:pospres}
		Assume that $\HS=L^2(\sM,\nu)$ is the space of square-integrable functions over a measure space $(\sM,\nu)$
		and that the semigroup $(e^{-tH})_{t\geq 0}$ is positivity preserving,
		i.e.\ that $e^{-tH}\phi\ge 0$ holds $\nu$-almost everywhere for any $\phi\in L^2(\cM,\nu)$ such that $\phi\ge0$ holds $\nu$-a.e. Then
		\cref{eq:infspec} holds for any $\psi\in L^2(\cM,\nu)$ such that $\psi>0$ $\nu$-a.e, see \cite[Thm.~C.1]{MatteMoller.2018} for a detailed proof.
		We will apply this in \cref{Section: Generalized Spin-Boson models}, by using that the operator of interest therein is unitarily equivalent to an operator $H$ of that form. 
	\end{rem}
	Let us define the family of rank-one perturbations
	\begin{equation*}
		H_\alpha \coloneqq H + \alpha \langle \psi, \cdot \rangle \psi
	\end{equation*}
	and let $E_\alpha \coloneqq \inf \sigma(H_\alpha)$. Using that 
	\begin{align}
		\label{eq:variation}
		E_\alpha = \inf \big \{\langle \phi, H_\alpha\phi \rangle:\, \phi \in \mathcal D(H), \, \|\phi\| = 1   \big\},
	\end{align}
	we observe that the function $\IR\to\IR,\,\alpha \mapsto E_\alpha$ is increasing as well as concave,
	as the infimum of increasing and concave (in fact linear) functions.
	
	Formally, when assuming that $E_\alpha$ is an eigenvalue of $H_\alpha$ and that both $E_\alpha$ and the eigenstates $\phi_\alpha$ can be developed into a series expansion,
	a simple coefficient comparison suggests that
	\begin{equation*}
		\partial_\alpha E_{\alpha}|_{\alpha = 0} = |\langle \psi,\phi_0 \rangle |^2 = \mu(\{E\}), \quad \quad -\partial_\alpha^2 E_\alpha|_{\alpha = 0} =
		 2 \mu(\{E\})  \int_{(E, \infty)} \frac{1}{x-E} \, \mu(\mathrm dx)
	\end{equation*}
	where $\phi_0$ is the ground state of $H_0$.
	These are the main formulas from perturbation theory, which we now want to connect with our main results \cref{Theorem: Convergence of averaged partition function,Theorem: Second order term}.
	Note that they do not immediately make sense, since $\alpha \mapsto E_\alpha$ may not even be differentiable.
	
	In the spirit of \cref{eq:infsupp}, let us thus replace $E_\alpha$ by the approximation
	\begin{equation*}
		E_{\alpha, t} \coloneqq -\frac{1}{t} \operatorname{log}\, \langle \psi, e^{-tH_\alpha} \psi \rangle 
	\end{equation*}
	for which we will prove in \cref{Lemma: Monotoicity in admissibilty} that for $\alpha\leq 0$ it also converges to $E_\alpha$ as $t\to \infty$.
	In our case, Duhamels formula (or alternatively a Dyson series expansion) simplifies to
	\begin{equation*}
		\partial_\alpha \langle \psi, e^{-tH_\alpha} \psi \rangle|_{\alpha = 0} = -\int_0^t Z_s Z_{t-s} \, \mathrm ds
	\end{equation*}
	(see \cref{eq:Dysonfirstorder} for details). Hence, \cref{Theorem: Convergence of averaged partition function} states exactly that
	\begin{equation*}
		\lim_{t\to \infty} \partial_\alpha E_{\alpha, t}|_{\alpha = 0} = \lim_{t\to \infty} \frac{1}{t}\int_0^t \frac{Z_{t-s} Z_s}{Z_t} \, \mathrm ds = \mu(\{E\}).
	\end{equation*}
	A similar reasoning can be applied for the second order.
	Expanding the Dyson series further yields (see \cref{eq:Dysonsecondorder} for details)
	\begin{equation*}
		\partial_\alpha^2 \langle \psi, e^{-tH_\alpha} \psi \rangle|_{\alpha = 0} = 2\int_0^t\int_0^s Z_{t-s} Z_{s-r} Z_r \, \mathrm dr \mathrm ds
	\end{equation*}
	and \cref{Theorem: Convergence of averaged partition function,Theorem: Second order term} thus imply that (under the additional assumptions of \cref{Theorem: Second order term})
	\begin{align*}
		-\lim_{t\to \infty} \partial_\alpha^2 E_{\alpha, t}|_{\alpha = 0} &= \lim_{t\to \infty} \frac{2}{t}\int_0^t\int_0^s \frac{Z_{t-s} Z_{s-r} Z_r}{Z_t} \, \mathrm dr \mathrm ds - \frac{1}{t}\Big(\int_0^t \frac{Z_s Z_{t-s}}{Z_t} \, \mathrm ds\Big)^2 \\
		&= 2\mu(\{E\})  \int_{(E, \infty)} \frac{1}{x-E} \, \mu(\mathrm dx).
	\end{align*}
	Below, we will further apply \cref{Theorem: Convergence of averaged partition function} in order to derive additional properties of the bottom of the spectrum of the rank-one perturbed operator $H_\alpha$.
	A result of this type in the case that $E_\alpha$ is a non-degenerate discrete eigenvalue of $H_\alpha$ is known as a simple case of the Feynman--Hellmann Theorem.
	We here emphasize that for our result neither the a priori assumption of $E_\alpha$ being an eigenvalue nor the differentiability of $\alpha\mapsto E_\alpha$ are required.
	Denoting by $\mu_\alpha$, $\alpha\in\IR$ the spectral measure of $H_\alpha$ with respect to the vector $\psi$ and by showing that $\alpha \mapsto E_{\alpha, t}$ is  concave for every $t\geq 0$, we will obtain
	\begin{thm}
		\label{Theorem: rank-one perturbation}
		Let $\partial_\alpha^-E, \partial_\alpha^+E$ denote the left and right derivatives of $\alpha \mapsto E_\alpha$ respectively (which exist by concavity). Then for all $\alpha \leq 0$
		\begin{equation*}
			\partial_\alpha^+ E_{\alpha} \leq \mu_\alpha(\{E_\alpha\}) \leq 	\partial_\alpha^- E_{\alpha} .
		\end{equation*}
		Moreover, $\alpha \mapsto \mu_\alpha(\{E_\alpha\})$ is a decreasing and left continuous function on $(-\infty, 0]$.
	\end{thm}
	
	\subsection{Admissibility of Coupling Constants}
	
	Let us first discuss validity of 
	\begin{align}
		\label{eq:convapproxgsenerg}
		\inf \sigma(H_\alpha)=E_\alpha = \lim_{t \to \infty} E_{\alpha,t}.
	\end{align}
	By \cref{eq:infsupp}, this is equivalent to the following criterion
	\begin{defn}
		We call $\alpha \in \R$ admissible if $\inf \sigma(H_\alpha)=\inf\supp \mu_\alpha$ holds. 
	\end{defn}
	We make the following two simple observations.
	\begin{lem}
		\label{Lemma: Monotoicity in admissibilty}
		All $\alpha \leq 0$ are admissible.
	\end{lem}
	\begin{proof}
		By assumption, we have that $\alpha = 0$ is admissible,
		so from now	assume $\alpha < 0$.
		From \cref{eq:variation}, we then have $E_\alpha \leq E_0 = E$.
		If $E_\alpha < E$, then it follows that $E_\alpha$ is an eigenvalue of finite multiplicity of $H_\alpha$, since finite rank perturbations keep the essential spectrum invariant.
		In this case any corresponding eigenvector $\phi$ has to satisfy $\braket{\phi,\psi}\ne 0$,
		as otherwise it would be an eigenvector of $H$ to the eigenvalue $E_\alpha$ as well, contradicting $E_\alpha < E$.
		Hence, in this case $\alpha$ is admissible.
		Now assume $E_\alpha = E$.
		Then for any $\varepsilon>0$ there exist a unit vector $\phi$ such that $\langle \phi, H_\alpha \phi \rangle \leq \langle \phi, H\phi \rangle \leq E + \varepsilon$ and $\langle \phi, \psi \rangle \neq 0$ (otherwise $\inf\supp\mu_0>E$ would hold). Hence $\inf\supp\mu_\alpha\le E+\varepsilon$ and taking $\eps\downarrow 0$ proves the statement.
	\end{proof}
	\begin{lem}
		\label{Lemma: Admissibility right side derivative}
		If $\partial_\alpha^+ E_\alpha|_{\alpha=0} > 0$ then there exists some $\varepsilon>0$ which is admissible.
	\end{lem}
	\begin{proof}
		By assumption we have $E_{\varepsilon} > E$ for all sufficiently small $\varepsilon>0$.
		Hence, again using that finite rank perturbations preserve the essential spectrum,
		$E$ is an isolated eigenvalue of $H$.
		Thus, taking $\eps>0$ sufficiently small, we find $E_{\varepsilon}\notin\sigma(H)$, 
		so in this case $E_{\varepsilon}$ is an isolated eigenvalue of finite multiplicity of $H_{\varepsilon}$
		and any corresponding eigenvector of $H_{\varepsilon}$ can not be orthogonal to $\psi$, since else $E_{\varepsilon}$ would be an eigenvalue of $H$ as well.
	\end{proof}
	
	\subsection{Proof of \cref{Theorem: rank-one perturbation}}
	\label{Section: Rank one perturbations}
	We start by proving the claimed concavity of $E_{\alpha,t}$.
	\begin{prop}
		\label{Proposition: Quotient of Partition functions in terms of rank one perturbations}
		The function $\alpha \mapsto E_{\alpha, t}$ is concave for any $t\geq 0$.
	\end{prop}
	\begin{proof}
		We now write
		\begin{align*}
			Z_{\alpha,t} \coloneqq \braket{\psi,e^{-tH_\alpha}\psi}.
		\end{align*}
		Applying a Dyson series expansion in $\alpha$, see for example \cite[Thm.~1.10]{EngelNagel.2000}, we obtain the convergent series
		\begin{align*}
			\braket{\psi,e^{-tH_\beta}\psi} = Z_t + \sum_{k=1}^{\infty}(\alpha-\beta)^k\int_0^t\int_0^{s_1}\cdots \int_0^{s_{k-1}}Z_{\alpha,t-s_1}Z_{\alpha,s_1-s_2}\cdots Z_{\alpha,s_{k-1}-s_k} Z_{\alpha, s_k}\,\rmd s_k\cdots\rmd s_1,
		\end{align*}
		Thus
		\begin{align}
			\frac{\mathrm d}{\mathrm d\alpha}  \langle \Omega, e^{-t H_{\alpha}} \Omega\rangle &= - \int_0^t Z_{\alpha, t-s} Z_{\alpha, s} \, \mathrm ds, \label{eq:Dysonfirstorder}\\
			\frac{\mathrm d^2}{\mathrm d\alpha^2}  \langle \Omega, e^{-t H_{\alpha}} \Omega\rangle &= 2\int_0^t \mathrm ds \,  Z_{\alpha, t-s} \int_0^s \mathrm dr \, Z_{\alpha, s-r} Z_{\alpha, r}, \label{eq:Dysonsecondorder}
		\end{align}
		which leads to
		\begin{equation}
			\label{eq:seconderestimate}
			-\frac{\mathrm d^2}{\mathrm d\alpha^2} E_{\alpha, t} = t^{-1}Z_{\alpha, t}^{-2} \cdot  \Big[2Z_{\alpha, t}\int_0^t \mathrm ds \,  Z_{\alpha, t-s} \int_0^s \mathrm dr \, Z_{\alpha, s-r} Z_{\alpha, r} - \Big(\int_0^t\mathrm ds \, Z_{\alpha, t-s} Z_{\alpha, s} \Big)^2 \Big].
		\end{equation}
		We now have
		\begin{align*}
			\Big(\int_0^t\mathrm ds \, Z_{\alpha, t-s} Z_s \Big)^2 &= \int_0^t \mathrm ds \int_0^t \mathrm dr \, Z_{\alpha, t-s} Z_{\alpha, s} Z_{\alpha, t-r} Z_{\alpha, r} \\
			&= 2 \int_0^t \mathrm ds \int_0^s \mathrm dr \, Z_{\alpha, t-s} Z_{\alpha, s} Z_{\alpha, t-r} Z_{\alpha, r} \\
			&\leq 2Z_{\alpha, t}\int_0^t \mathrm ds \,  Z_{\alpha, t-s} \int_0^s \mathrm dr  \, Z_{\alpha,s-r} Z_{\alpha,r} 
		\end{align*}
		where we used in the last inequality that by the second point of \cref{Proposition: Monotonicty of fraction of partition functions} the inequalities
		\begin{equation*}
			\frac{Z_{\alpha,t-r}}{Z_{\alpha, t}} \leq \frac{Z_{\alpha,s-r}}{Z_{\alpha,s}}   \quad \iff\quad Z_{\alpha,s} Z_{\alpha,t-r} \leq Z_{\alpha, t} Z_{\alpha, s-r}.
		\end{equation*}
		hold for all $0 \leq r\leq s \leq t$. Inserting the estimate into \cref{eq:seconderestimate} proves the statement.
	\end{proof}
	\begin{proof}[Proof of \cref{Theorem: rank-one perturbation}]
		By \cref{Lemma: Monotoicity in admissibilty}, we have $\lim_{t\to \infty} E_{\alpha, t} = E_\alpha$ for all $\alpha < 0$, which by \cref{Proposition: Quotient of Partition functions in terms of rank one perturbations,Theorem: Convergence of averaged partition function} implies
		\begin{equation*}
			\partial^+_\alpha E_\alpha \leq \lim_{t\to \infty}\partial_\alpha E_{\alpha, t} = \mu_\alpha(\{E_\alpha\}) \leq  \partial^-_\alpha E_\alpha.
		\end{equation*}
		For $\alpha=0$, the upper inequality follow by the same argument. The lower inequality trivially holds if $\partial^+_\alpha E_\alpha= 0$ and if $\partial^+_\alpha E_\alpha> 0$ one can apply \cref{Lemma: Admissibility right side derivative} to obtain convergence of $E_{\alpha, t}$ to $E_\alpha$ on an interval of the form $(-\infty, \varepsilon)$ and hence the same argument applies.
		
		Since
		\begin{equation*}
			\frac{\mathrm d}{\mathrm d \alpha} \frac{1}{t} \int_0^t \frac{Z_{\alpha, s} Z_{\alpha, t-s}}{Z_{\alpha, t}} \mathrm ds = \partial^2_\alpha E_{\alpha, t}\leq 0
		\end{equation*}
		the function $\alpha \mapsto \mu_\alpha(\{E_{\alpha}\})$ is decreasing as a pointwise limit of decreasing functions.
		
		For the left continuity, it is sufficient to notice that for any $\alpha_0 \leq 0$
		\begin{equation*}
			\mu_\alpha(\{E_{\alpha}\}) = \lim_{t \to \infty} \lim_{\alpha \uparrow \alpha_0} \frac{1}{t} \int_0^t \frac{Z_{\alpha, s} Z_{\alpha, t-s}}{Z_{\alpha, t}} \mathrm ds =  \lim_{\alpha \uparrow \alpha_0} \rho(\alpha)
		\end{equation*}
		as we might exchange the order of limits as the expression is both decreasing in $\alpha$ (by the above) as well as in $t$, by \cref{Proposition: Monotonicty of fraction of partition functions}.
	\end{proof}

	\section{A link to renewal theory}
	\label{Section: Renewal theory}
	In this section, we give a natural probabilistic interpretation of our two Wiener-type theorems.
	To do this, we assign to each probability measure $\mu$ (or better, to its equivalence class modulo translations) with finite mean and lower bounded support a $\{0, 1\}$-valued regenerative stochastic process.
	We will call the latter the renewal transform of the measure $\mu$ and will show that it uniquely determines $\mu$ up to translations.
	The renewal transform will allow us to give an intuitive interpretation of \cref{Theorem: Convergence of averaged partition function,Theorem: Second order term}.
	As we will point out in \cref{Example: Renewal transform of the Fröhlich Polaron}, for the spectral measure of the Fr\"ohlich polaron (taken with respect to the Fock vacuum),
	the renewal transform can be expressed in terms of its point process representation which was first introduced in \cite{MukherjeeVaradhan.2020}.
	The latter has then successfully been applied in order to study the ground state energy $E(P)$ of the Hamiltonian $H(P)$ at fixed total momentum $P$ and, in particular, the effective mass, i.e., the curvature of $P \mapsto E(P)$ in the origin, see \cite{BetzPolzer.2023,Polzer.2023,BazaesMukherjeeSellkeVaradhan.2023}.
	The interpretation of \cref{Theorem: Second order term} in terms of the renewal transform (see \cref{Theorem: Renewal representation Theorem 1 and Theorem 2})
	shows that one can in principle study the spectral measure of $H(P)$ even above its bottom by similar techniques.
	Considering the generality of our setup, identifying and analyzing the renewal transform might be a promising approach for the study of other quantum mechanical models as well.
	Moreover, it might be interesting to further exploit the connection to perturbation theory made in \cref{Section: Perturbation Theory} in order to derive similar expressions corresponding to higher order perturbation theory.
	We will first state the main results of this section and prove them afterwards, so that the reader can skip the proofs.
	
	As in \cref{Section: Wiener Type Theorems} let $\mu$ be a Borel probability measure on $\IR$ with support bounded from below and
	further assume that $\mu$ has a finite first moment $m\coloneqq\int_\IR x\,\mu(\rmd x)$.
	We will assign a stochastic process $X$ to $\mu$ that alternates between two states, dormant and active, and which regenerates after each cycle consisting of a dormant period followed by an active period.
	Let $\mathbf P$ be a probability measure on the space
	\begin{equation*}
		\mathcal D \coloneqq \{x:[0, \infty) \to \{0, 1\}:\, x \text{ is càdlàg and }x_0 = 0\}
	\end{equation*}
	which we equip, as usual, with the $\sigma$-algebra generated by the evaluation maps $x\mapsto X_t(x) \coloneqq x_t, \ t\geq 0$. We denote by $X \coloneqq (X_t)_{t\geq 0}$ the canonical stochastic process with law $\mathbf P$. The process $X$ partitions the half line $[0, \infty)$ into dormant periods, in which $X_t = 0,$ and active periods, in which $X_t = 1$. We denote by
	\begin{equation*}
		d_1 \coloneqq \inf\{t\geq 0: X_t = 1\}, \quad a_1 \coloneqq 
		\inf\{t-d_1: t\geq d_1,\, X_t = 0\}
	\end{equation*}
	the first dormant and the first active period (which might be infinite) and set $T_1 \coloneqq d_1 + a_1$ to be the first return to 0. We call $X$ an alternating renewal process (with respect to $\mathbf P$) if either $d_1 = \infty$ almost surely or if $(X_{t - T_1})_{t\geq T_1}$ is conditionally on the event $\{T_1< \infty\}$ (provided the latter has positive probability) independent of $(d_1, a_1)$ and has law $\mathbf P$.\footnote{If $d_1$ and $a_1$ are almost surely finite, this agrees with the common definition of an alternating renewal process. However, we allow $T_1$ to be infinite with positive probability and hence allow the embedded renewal process to die out.} Notice that $\mathbf P$ is then uniquely determined by the law of $(d_1, a_1)$ under $\mathbf P$, as the successive cycles, consisting each of a dormant period followed by an active period, are independent of each other. Provided that $t\geq 0$ is such that $\mathbf P(X_t = 0)>0$, we define
	\begin{equation*}
		\mathbf P_t(\, \cdot \,) \coloneqq \mathbf P(\, \cdot \,|X_t = 0)
	\end{equation*}
	to be the law of $X$ conditional on $X_t = 0$ and denote by $\mathbf E_t$ and $\mathbf V_t$ the expected value and variance taken with respect to $\mathbf P_t$. 
	To simplify notation, we denote by $\operatorname{Exp}(0)$ the law of a random variable which is almost surely $+\infty$.
	\begin{thm}
		\label{Theorem: Existence renewal transform}
		There exists a unique probability measure $\mathbf P$ on $\mathcal D$ such that $X$ is an alternating renewal process, such that $d_1$ is exponentially distributed and independent of $a_1$ and such that
		\begin{equation*}
			e^{Et} Z_t = \mathbf P(X_t = 0)
		\end{equation*}
		for all $t\geq 0$. We have $d_1 \sim \operatorname{Exp}(m-E)$ under $\mathbf P$ and for all $0= t_0 \leq t_1 \leq \hdots \leq t_n = t$
		\begin{equation*}
			\frac{1}{Z_t}\prod_{i=0}^{n-1} Z_{t_{i+1}-t_i} = \mathbf P_t\big(X_{t_1} = X_{t_2} = \hdots = X_{t_{n-1}} = 0\big).
		\end{equation*}		 
	\end{thm}
	We call $\mathbf P$ the renewal-transform of $\mu$. Notice that the renewal transform only uniquely determines the measure up to translations.
	We will explicitly construct the measure $\mathbf P$ in the proof of \cref{Theorem: Existence renewal transform} below in terms of the alternating idle and busy periods of a $M/G/\infty$ queue.
	That being said, for applying our theory to spectral measures it might be preferable to identify the renewal transform by other means such as Feynman--Kac formulas, see \cref{Example: Renewal transform of a random walk,Example: Renewal transform of discrete Schroedinger operator,Example: Renewal transform of the Fröhlich Polaron} below.
	By applying \cref{Theorem: Existence renewal transform} in order to derive a renewal equation for the Laplace transform, we can also express the Stieltjes transform on the half-plane $\{z\in \C:\,\operatorname{Re}(z)<E\}$ in terms of the renewal transform.
	
	\begin{prop}
		\label{Proposition: Stieltjes transform in terms of renewal transform}
		For every $z\in \mathbb C$ with $\operatorname{Re}(z)<E$, we have
		\begin{equation*}
			\int_{[E, \infty)} \frac{\mu(\mathrm dx)}{x-z} = \frac{1}{m - z} \cdot \frac{1}{1 - \mathbf E[e^{(z - E)T_1} \1_{\{T_1 < \infty\}}]}.
		\end{equation*}
	\end{prop}
	We will also prove \cref{Proposition: Stieltjes transform in terms of renewal transform} at the end of this section. 
	
	Let us assume for the moment that $\mu$ is not a Dirac measure such that that $m-E>0$. Let $D_t \coloneqq \int_0^t (1-X_s) \, \mathrm ds$
	be the total dormant time up to time $t$.
	We then have by standard renewal theoretic arguments (which we will summarize at the end of this section)
	\begin{equation}
		\label{Equation: Convergence of dormant probability to atom}
		\mu(\{E\}) = \lim_{t\to \infty} e^{Et} Z_t = \lim_{t\to \infty}\mathbf P(X_t = 0) =  \lim_{t\to \infty} \mathbf E[D_t/t] = \frac{\mathbf E[d_1]}{\mathbf E[T_1]} = \frac{1}{1 + (m-E)\mathbf E[a_1]}
	\end{equation}
	where we used in the last equality that $d_1 \sim \operatorname{Exp}(m-E)$, and where the last two expressions are by definition zero in case that $\mathbf E[T_1]  = \mathbf E[a_1] = \infty$. In combination with \cref{Proposition: Stieltjes transform in terms of renewal transform} this leads to the following observation
	\begin{cor}
		\label{Corollary: classify singularity in E}
		We have
		\begin{align*}
			&\mu(\{E\}) = 1 & &\iff \quad  d_1 = \infty \text{ a.s. } \\
			&\mu(\{E\}) \in (0, 1)&  &\iff \quad \mathbf E[T_1] < \infty. \\
			&\mu(\{E\}) = 0 \text{ and } \int_{(E, \infty)} \frac{\mu(\mathrm dx)}{x-E} = \infty & &\iff \quad T_1< \infty \text{ a.s. and }\mathbf E[a_1] = \infty. \\
			&\mu(\{E\}) = 0 \text{ and } \int_{(E, \infty)} \frac{\mu(\mathrm dx)}{x-E} < \infty & &\iff \quad d_1< \infty \text{ a.s. and }\mathbf P(a_1 =  \infty) > 0.	
		\end{align*}
	\end{cor}
	
	For a visualization of \cref{Corollary: classify singularity in E} see \cref{fig:alternatingrenewal2}.
	Notice that $p \coloneqq \mathbf P(a_1 = \infty)>0$ implies that the total number of dormant periods has geometric distribution with success probability $p$ and has therefore mean $1/p$.
	As a second consequence of \cref{Proposition: Stieltjes transform in terms of renewal transform}, we hence obtain the following.
	\begin{cor}
	If $\mu(\{E\}) = 0$ then $(m-E) \int_{(E, \infty)} (x-E)^{-1}\mu(\mathrm dx)$
	is the expected total number of dormant periods under $\mathbf P$.
	\end{cor} 
	\begin{figure}
		\centering
		\includegraphics[width=0.8\linewidth]{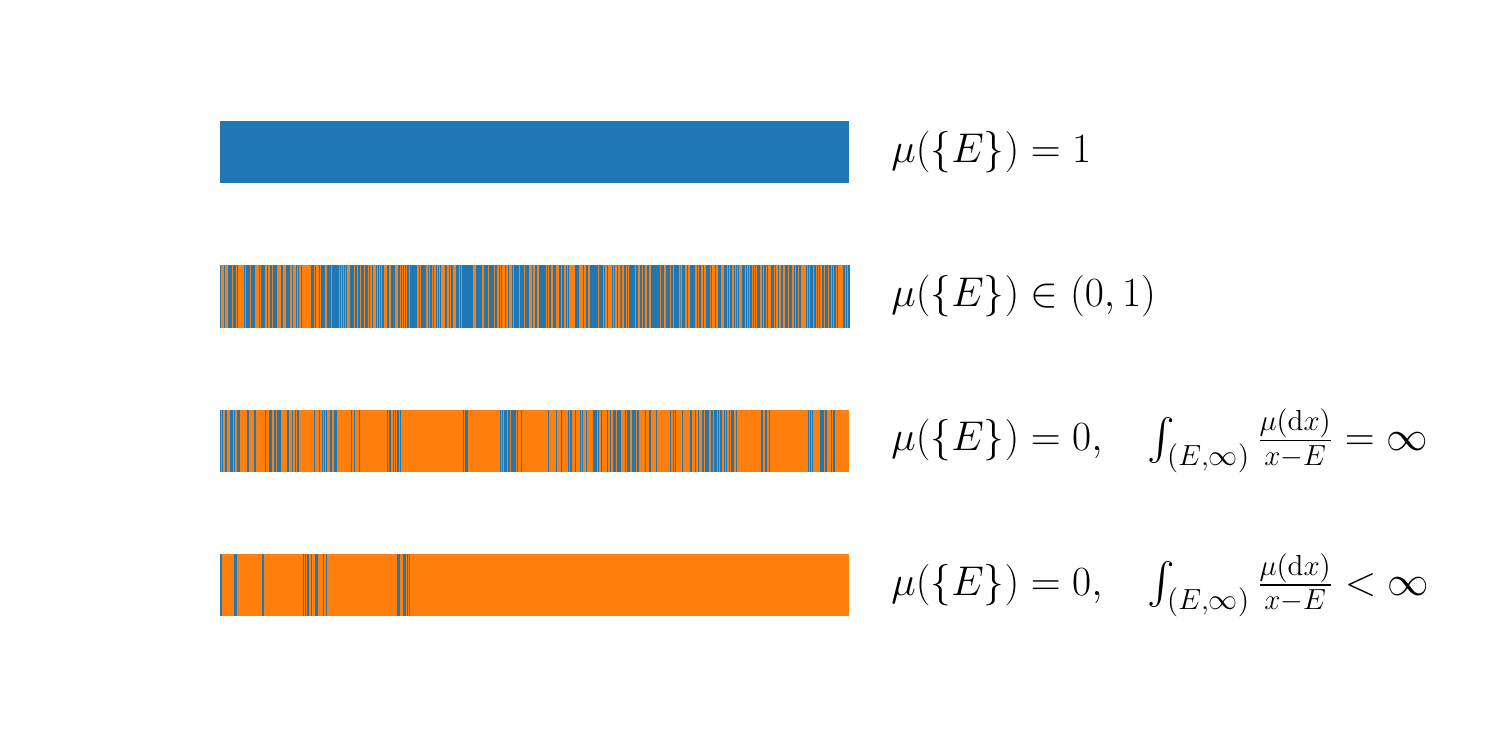}
		\vspace*{-7mm} 
		\caption{Visualization of \cref{Corollary: classify singularity in E}. Dormant periods are in blue, active periods are in orange. The more mass the measure has close to $E = \inf \operatorname{supp}(\mu)$, the more the process tends to be dormant. For $\mu = \delta_E$ the process is always dormant, for $\mu(\{E\}) = 0$ and $\int_{(E, \infty)} (x-E)^{-1} \mu(\mathrm dx) < \infty$ the process is eventually active.}
		\label{fig:alternatingrenewal2}
	\end{figure}
	Let us now rephrase \cref{Theorem: Convergence of averaged partition function,Theorem: Second order term} in terms of the renewal transform. By \cref{Theorem: Existence renewal transform}, we have for any $t>0$
	\begin{equation*}
		\int_0^t \frac{Z_s Z_{t-s}}{Z_t} \, \mathrm ds =  \int_0^t \mathbf P_t(X_s = 0) \, \mathrm ds = \mathbf E_t[D_t],
	\end{equation*}
	where the last equality follows from Fubini's theorem.
	Hence, \cref{Theorem: Convergence of averaged partition function} states exactly that for all $\kappa\in (0, 1)$
	\begin{equation*}
		\mu(\{E\}) = \lim_{t\to \infty} \mathbf P_t(X_{\kappa t} = 0) = \lim_{t\to \infty} \mathbf E_{t}\big[D_t/t \big].
	\end{equation*}
	Comparing with \cref{Equation: Convergence of dormant probability to atom} yields to the following intuitive interpretation of \cref{Theorem: Convergence of averaged partition function}: the latter states exactly that
	\begin{equation}
		\label{Equation: Conditioning does not change value of limit}
		\lim_{t\to \infty} \mathbf P_t(X_{\kappa t} = 0) = \lim_{t\to \infty} \mathbf P(X_{\kappa t} = 0), \quad  \lim_{t\to \infty} \mathbf E_{t}\big[D_t/t \big] =  \lim_{t\to \infty} \mathbf E\big[D_t/t \big],
	\end{equation}
	i.e., that we do not change the value of the limit by conditioning on $\{X_t = 0\}$. For the case where $T_1$ has finite expected value (corresponding to the case where $\mu(\{E\})>0$), one can directly show \cref{Equation: Conditioning does not change value of limit} via renewal theory. Notice, that \cref{Equation: Conditioning does not change value of limit} in general does not hold for arbitrary alternating renewal processes which are not the renewal transform of a probability measure: take, for example, the case where the dormant periods are exponentially distributed and the distribution of the first active period is of the form $p\delta_\infty + (1-p)\delta_c$ for some $p\in (0, 1)$ and $c>0$, and apply \cite[Proposition 4.7]{BetzPolzer.2022}.
	
	In the same manner, we have
	\begin{equation*}
		2\int_0^t \mathrm ds \int_0^s \mathrm dr \, \cfrac{Z_{t-s} Z_{s-r} Z_r}{Z_t}
		=\mathbf E_t\Big[\int_0^t \mathrm ds \int_0^t \mathrm dr \, \1_{\{X_s = X_r = 0\}} \Big] = \mathbf E_t[D_t^2]
	\end{equation*}
	such that
	\begin{equation*}
		2\int_0^t \mathrm ds \int_0^s \mathrm dr \, \cfrac{Z_{t-s} Z_{s-r} Z_r}{Z_t} -\Big( \int_0^t \, \mathrm ds\cfrac{Z_{t-s} Z_s}{Z_t}\Big)^2 = \mathbf V_t[D_t].
	\end{equation*}
	Hence, \cref{Theorem: Second order term} exactly states that if $\mu(\{E\})>0$ then
	\begin{equation}
		\label{Equation: Theorem 2 in terms of renewal transform}
		\int_{(E, \infty)} \frac{\mu(\mathrm dx)}{x-E} = \lim_{t\to \infty} \frac{\mathbf V_t[D_t]}{\mathbf E_t[D_t]}
	\end{equation}
	provided that the left hand side is finite. We will show below that \cref{Proposition: Stieltjes transform in terms of renewal transform} implies that
	\begin{equation}
		\label{Equation: Convergence of variance uncondtioned}
		\int_{(E, \infty)}\frac{\mu(\mathrm dx)}{x-E} = \frac{(m-E)\mathbf E[a_1^2]}{\big(1+ (m-E)\mathbf E[a_1]\big)^2} =  \lim_{t\to \infty} \frac{\mathbf V[D_t]}{\mathbf E[D_t]},
	\end{equation}
	where the second equality will follow from the known asymptotic variance of renewal-reward processes.
	Comparing \cref{Equation: Theorem 2 in terms of renewal transform,Equation: Convergence of variance uncondtioned} and taking \cref{Equation: Conditioning does not change value of limit} into account, \cref{Theorem: Second order term} can be restated as
	\begin{equation*}
		\lim_{t\to \infty} \mathbf V_t[D_t/t] = \lim_{t\to \infty} \mathbf V[D_t/t]
	\end{equation*}
	i.e.\ once more we do not change the value of the limit by conditioning on $\{X_t = 0\}$. Let us summarize the above in  the following theorem.
	\begin{thm}
		\label{Theorem: Renewal representation Theorem 1 and Theorem 2}
		We have for every $\kappa \in (0, 1)$
		\begin{equation*}
			\mu(\{E\}) = \lim_{t\to \infty} \mathbf P(X_t = 0) = \mathbf P_t(X_{\kappa t} = 0) = \lim_{t\to \infty} \mathbf E[D_t/t] =  \lim_{t\to \infty} \mathbf E_t[D_t/t] = \frac{1}{1+(m-E)\mathbf E[a_1]}.
		\end{equation*}
		Moreover, if $\mu(\{E\})>0$ holds then
		\begin{equation*}
			2\int_{(E, \infty)}\frac{\mu(\mathrm dx)}{x-E}  = \lim_{t\to \infty} \frac{\mathbf V[D_t]}{\mathbf E[D_t]} = \lim_{t\to \infty} \frac{\mathbf V_t[D_t]}{\mathbf E_t[D_t]} = \frac{(m-E)\mathbf E[a_1^2]}{\big(1+ (m-E)\mathbf E[a_1]\big)^2}
		\end{equation*}
		provided that the left hand side is finite or, equivalently, provided that $\mathbf E[a_1^2]<\infty$.  
	\end{thm}
	In other words, if we assume that $\mu$ is not a Dirac measure then we have $\mu(\{E\})>0$ if and only if $T_1$ has finite expected value. In this case, a stationary version of the process exists which is the limit of the distribution of $(X_t)_{t\geq T}$ as $T\to \infty$ (see e.g.\ \cite[Proposition 4.4]{BetzPolzer.2022}) and we have $\int_{(E, \infty)} (x-E)^{-1} \mu(\mathrm dx)< \infty$ if and only if the time until the first renewal of the stationary process has finite expected value.
	We will finish the proof of \cref{Theorem: Renewal representation Theorem 1 and Theorem 2} at the end of this section by proving \cref{Equation: Convergence of variance uncondtioned}.
	Before looking at some concrete examples of renewal transforms, we point out that we obtain as a corollary of \cref{Proposition: Monotonicty of fraction of partition functions} some potentially useful monotonicity properties.
	
	\begin{cor}	
		The following holds.
		\begin{enumerate}
			\item For every $t\geq 0$ the function $s \mapsto \mathbf P_t(X_s = 0)$ is decreasing on $[0, t/2]$ and increasing on $[t/2, t]$.
			\item For every $s\geq 0$ the function $t \mapsto \mathbf P_t(X_s = 0)$ is decreasing on $[s, \infty)$.
			\item The function $t\mapsto \mathbf E_t[D_t/t]$ is decreasing on $[0, \infty)$.
		\end{enumerate}
	\end{cor}

	While we will explicitly construct the process $\mathbf P$ in the proof of \cref{Theorem: Existence renewal transform} as the idle and busy periods of a $M/G/\infty$-queue, we will see in the following examples that there can be different ways to realize the renewal transform, such as Feynman--Kac formulas. In this context, we will see that it can be easier to work with the conditioned process $\mathbf P_t$ on the interval $[0, t]$ (which, in the given examples, can be expressed in terms of a perturbed path measure in finite volume) than to work with the full measure $\mathbf P$ (which can be expressed in terms of the infinite volume limit of the perturbed path measure provided that it exists).
	
	\begin{ex}
		\label{Example: Renewal transform of a random walk}
		Let $\Delta$ denote the discrete Laplace operator on $\ell^2(\Z^d)$, meaning that
		\begin{equation*}
			(\Delta \psi)(x) \coloneqq \sum_{y:\, \|x-y\|_1 = 1} (\psi(y) - \psi(x)),
		\end{equation*}
		for all $\psi\in \ell^2(\mathbb Z^d)$ and $x\in \mathbb Z^d$. Let us fix some arbitrary vertex $o\in \mathbb Z^d$ and let $\mu$ be the spectral measure of $-\Delta$ with respect to the unit vector $\delta_o \coloneqq (\delta_{ox})_{x \in \Z^d}$. 
		Let $(Y_t)_{t\geq 0}$ be a simple continuous time random walk on $\mathbb Z^d$ started in $o$, whose distribution we will denote by $\mathbb P$, and define for $t\geq 0$
		\begin{equation*}
			\hat X_t \coloneqq \1_{\{Y_t \neq  o\}}.
		\end{equation*} 
		Since $\Delta$ is the generator of $Y$, we have
		\begin{equation*}
			Z_t = \langle \delta_o, e^{t\Delta} \delta_o \rangle = \mathbb E[\1_{\{Y_t = o\}}] = \mathbb P(\hat X_t = 0)
		\end{equation*}
		and since $E=0$, one easily sees that the law of $(\hat X_t)_{t\geq 0}$ is the renewal transform of $\mu$. Let $\hat T_1$ be the first recurrence time of $Y$ to $o$ (which we set to be $\infty$ in case the walk does not return). Then $\mathbb P(\hat T_1 = \infty) >0$ if and only if $d\geq 3$.
		Since $\mu(\{0\}) =0$, we obtain with \cref{Corollary: classify singularity in E}
		\begin{equation*}
			\int_{(0, \infty)} \frac{\mu(\mathrm dx)}{x} < \infty  \quad \iff \quad  d = 1, 2.
		\end{equation*}
		This agrees with the known asymptotics of the density $\rho$ of $\mu$ as $x\downarrow 0$: It is well known that $\rho(x)\sim x^{d/2-1}$ as $x\downarrow 0$ (see e.g. \cite[Exercise 4.2]{AizenmanWarzel.2015}).
	\end{ex}		
	
	\begin{ex}
		\label{Example: Renewal transform of discrete Schroedinger operator}
		More generally, let $H = -\Delta + V$ be the discrete Laplace operator with potential $V = (V(x))_{x\in \mathbb Z^d}$ (acting as a multiplication operator on $\ell^2(\mathbb Z^d)$), which we assume, for simplicity, to be bounded and let $\mu$ be the spectral measure with respect to $\delta_o$. 
		Then the Feynman--Kac formula yields
		\begin{equation*}
			Z_t = \langle \delta_o, e^{-tH} \delta_o \rangle = \mathbb E\Big[\operatorname{exp}\Big(-\int_0^t V(Y_s) \, \mathrm ds \Big) \1_{\{Y_t = o\}} \Big].
		\end{equation*}
		Define the perturbed random walk
		\begin{equation*}
			\widehat{\mathbb P_t}(\mathrm dY) \coloneqq \frac{1}{Z_t}\operatorname{exp}\Big( -\int_0^t V(Y_s)\, \mathrm ds\Big)  \1_{\{Y_t = o\}} \, \mathbb P(\mathrm dY)
		\end{equation*}
		where $Z_t$ is a normalization constant.
		As above, let $\hat X_t \coloneqq \1_{\{Y_t \neq  o\}}$. One easily shows that for all $0\leq r\leq s\leq t$
		\begin{equation*}
			\mathbf P_t(X_s = 0) = \frac{Z_s Z_{t-s}}{Z_t} = \widehat{\mathbb P}_t(\hat X_s = 0), \quad 	\mathbf P_t(X_r = X_s = 0) = \frac{Z_r Z_{s-r} Z_{t-s}}{Z_t} =   \widehat{\mathbb P}_t(\hat X_r = \hat X_s = 0).
		\end{equation*}
		Similarily to the argument in the proof of \cref{Proposition: Stieltjes transform in terms of renewal transform}, namely by deriving a renewal equation for $Z$ and by taking the Laplace transform in said renewal equation, one obtains by comparison with \cref{Proposition: Stieltjes transform in terms of renewal transform} an explicit representation of the renwal transform of $\mu$: the measure
		\begin{equation*}
			\widetilde{\mathbb P}(\mathrm dX) \coloneqq \operatorname{exp}\Big( \int_0^{\hat T_1} (E-V(X_t)) \, \mathrm dt \Big) \1_{\{\hat T_1 < \infty\}} \, \mathbb P(\mathrm dX)
		\end{equation*}
		is a sub-probability measure whose total mass we denote by $p \in (0, 1]$. Let $\nu$ be the image measure of $\widetilde{\mathbb P}$ under the map $\hat T_1$. The renewal transform $\mathbf P$ is the distribution of the unique alternating renewal process under which $d_1$ and $a_1$ are independent, under which $d_1\sim \operatorname{Exp}(2d + V(o) - E)$ and under which $T_1$ has distribution $\nu + (1-p) \delta_\infty$.
		By \cref{Corollary: classify singularity in E}, the operator $H$ has an eigenfunction $\psi$ to the eigenvalue $E$ with $\psi(o) \neq 0$ if and only if $p=1$ and $\widehat{\mathbb E}[\hat T_1] < \infty$. In this case, the infinite volume limit $\widehat{\mathbb P} = \lim_{t \to \infty} \widehat{\mathbb P_t}$ exists in a suitable sense (this follows from \cite[Prop.~4.4]{BetzPolzer.2022}) and $\mathbf P$ is the distribution of $(\hat X_t)_{t\geq 0}$ under $\widehat{\mathbb P}$.
	\end{ex}

	\begin{ex}
		\label{Example: Renewal transform of the Fröhlich Polaron}
		For a more intricate example, let us consider the Hamiltonian $H(0)$ of the Fröhlich Polaron at fixed total momentum $0$ and coupling $\alpha>0$. Let $\mu$ be the spectral measure of $H(0)$ with respect to the Fock vacuum $\Omega$.
		Let $\widehat{\Gamma}_t$ be the dual point process of the path measure $\widehat{\mathbb P}_t$ of the Polaron in finite volume $[0, t]$, see \cite{MukherjeeVaradhan.2020,BetzPolzer.2022}. Then $\widehat{\Gamma}_t$ can be seen as the law of a perturbed $M/G/\infty$-queue, conditioned to be empty at time $t$. For $t\geq 0$, let $N_t$ denote the number of customers present at time $t$ and let
		\begin{equation*}
			\hat X_t \coloneqq \1_{\{N_t > 0\}}, \quad t\geq 0.
		\end{equation*}
		One can show that for all $0\leq r \leq s \leq t$
		\begin{equation*}
			\mathbf P_t(X_s = 0) = \frac{Z_s Z_{t-s}}{Z_t} = \widehat{\Gamma}_t(\hat X_s = 0), \quad 	\mathbf P_t(X_r = X_s = 0) = \frac{Z_r Z_{s-r} Z_{t-s}}{Z_t} =   \widehat{\Gamma}_t(\hat X_r = \hat X_s = 0).
		\end{equation*}
		Let $\widehat{\Gamma}$ be the infinite volume limit of $\widehat{\Gamma}_t$ as $t\to \infty$, which was shown to exist in \cite{MukherjeeVaradhan.2020,BetzPolzer.2022}.
		Then $\mathbf P$ is the distribution of $(\hat X_t)_{t\geq 0}$ under $\widehat{\Gamma}$; compare \cite[Prop.~3]{Polzer.2023} with \cref{Proposition: Stieltjes transform in terms of renewal transform}, taking into consideration that $m = \big\langle \Omega, H(0) \Omega \big\rangle = 0$.
		However, $\widehat{\Gamma}$ and the law $\mathbf Q$ of the $M/G/\infty$ queue $\xi$ constructed in the proof of \cref{Theorem: Existence renewal transform} do in general not coincide: Using that \cite{DonskerVaradhan.1983} $\operatorname{inf} \sigma(H(0)) \sim c \alpha^2$ as $\alpha\to \infty$ and \cite[Eq.~(7.1)]{BetzPolzer.2022}, one obtains that the density of individuals in the limit of large $\alpha$ is approximately twice as large under $\widehat{\Gamma}$ as under $\mathbf Q$. In other words, while the process $(\1_{\{N_t > 0\}})_{t\geq 0}$ has the same distribution under $\widehat{\Gamma}$ as under $\mathbf Q$, the process $(N_t)_{t\geq 0}$ of the number of customers in general does not.
	\end{ex}
	
	We conclude this section with the
	\subsection*{Proofs of \cref{Theorem: Existence renewal transform,Proposition: Stieltjes transform in terms of renewal transform,Equation: Convergence of dormant probability to atom,Theorem: Renewal representation Theorem 1 and Theorem 2}}
	
	\begin{proof}[Proof of \cref{Theorem: Existence renewal transform}]	
		We start by showing the existence of $\mathbf P$. After translation of $\mu$ (which leaves the function $t \mapsto e^{Et}Z_t$ invariant), we might assume with out loss of generality that $m=0$. Since the statement is trivial for the case where $\mu$ is a Dirac measure, we will assume w.l.o.g. that $E<m=0$.
		We define $\phi:[0, \infty) \to \R$ by $\phi(t) \coloneqq \log(Z_t)$ for all $t\geq 0$ (i.e.\ $t\mapsto \phi(-t)$ is the cumulant generating function). The function $\phi$ is differentiable on $[0, \infty)$ and twice differentiable on $(0, \infty)$, and one easily checks that
		\begin{equation*}
			\phi'(0) = -m = 0, \quad \phi''(t) = q(t) \coloneqq \mathbb V[\hat \mu_t]
		\end{equation*}
		where $\mathbb V[\hat \mu_t]$ denotes the variance of the probability measure $\hat \mu_t$ defined by
		\begin{equation*}
			\hat \mu_t(\mathrm dx) \coloneq \frac{1}{Z_t} e^{-tx} \mu(\mathrm dx).
		\end{equation*}
		Notice that
		\begin{align*}
			\operatorname{exp}\Big( \int_0^t \mathrm du \int_u^t \mathrm dv \, q(v-u) \Big) = \operatorname{exp}\Big( \int_0^t \mathrm du \, \phi'(t-u) \Big) = \exp(\phi(t)) = Z_t,
		\end{align*}
		since $\phi(0) = \phi'(0) = 0$. We have
		\begin{equation*}
			\int_0^t \mathrm du \int_u^t \mathrm dv\, q(v-u) = \int_0^t \mathrm du \int_0^{t-u} \mathrm d\tau \, q(\tau) = \int_0^t \mathrm d \tau q(\tau) (t- \tau)
		\end{equation*}
		leading to
		\begin{equation}
			\label{Equation: Obtain intensity from E}
			\lim_{t \to \infty}\int_0^t \mathrm d\tau\, q(\tau) (1-\tau/t) = \lim_{t\to \infty} \frac{1}{t} \log(Z_t) = -E.
		\end{equation}
		Since
		\begin{equation*}
			\frac{1}{2} \int_0^{t/2} \mathrm d \tau q(\tau) \leq \int_0^t \mathrm d \tau q(\tau) (1- \tau/t) , 
		\end{equation*}
		we obtain with \cref{Equation: Obtain intensity from E}
		\begin{equation*}
			\int_0^\infty \mathrm d\tau\, q(\tau) \leq -2E.
		\end{equation*}
		Hence, for every $\varepsilon>0$,
		\begin{equation*}
			\limsup_{t\to \infty} \frac{1}{t}\int_0^t \mathrm d\tau \, q(\tau) \tau \leq \limsup_{t\to \infty} \varepsilon \int_0^{\varepsilon t}  \mathrm d \tau \,  q(\tau) + \int_{\varepsilon t}^t \mathrm d \tau \, q(\tau) \leq -2\varepsilon E
		\end{equation*}
		and therefore  \cref{Equation: Obtain intensity from E} yields
		\begin{equation}
			\label{Equation: Arrival density is -E}
			\int_0^\infty \mathrm d\tau \, q(\tau) = -E.
		\end{equation}
		We can hence further rewrite
		\begin{align}
			\label{Equation: Z in terms of PPP}
			Z_t = \operatorname{exp}\Big( \int_0^t \mathrm du \int_u^t \mathrm dv \, q(v-u) \Big) &= \operatorname{exp}\Big( \int_0^t \mathrm du \int_u^\infty \mathrm dv \, q(v-u) - \int_0^t \mathrm du \int_t^\infty \mathrm d v \, q(v-u)   \Big) \nonumber \\
			&=\exp\Big(-Et - \int_0^t \mathrm du \int_t^\infty \mathrm d v \, q(v-u)\Big).
		\end{align}
		Let $\xi$ be a Poisson point process with intensity measure
		\begin{equation*}
			q(v-u) \1_{\{0 < u < v\}} \, \mathrm du \mathrm dv.
		\end{equation*}
		For $t\geq 0$, let
		\begin{equation*}
			N_t \coloneqq \xi\big([0, t] \times (t, \infty) \big) \sim \operatorname{Poi}\Big(\int_0^t \mathrm du \int_t^\infty \mathrm d v \, q(v-u)\Big) 
		\end{equation*}
		be the number of points of $\xi$ contained in the set $[0, t] \times (t, \infty)$, i.e.,\ the number of points $(u, v)$ of $\xi$ such that $u\leq t < v$. Then we can rewrite \cref{Equation: Z in terms of PPP} as
		\begin{equation*}
			Z_t = e^{-Et} \mathbb P(N_t =0) = e^{-Et} \mathbf P(X_t=0),
		\end{equation*}
		where we define $\mathbf P$ to be the distribution of the process $(\1_{\{N_t > 0\}})_{t\geq 0}$. Let $0= t_0 \leq t_1 \leq \hdots \leq t_n = t$.  We have
		\begin{equation*}
			Z_{t_{i+1}-t_i} = \operatorname{exp}\Big( \int_0^{t_{i+1}-t_i} \mathrm du \int_{u}^{t_{i+1}-t_i } \mathrm dv \, q(v-u) \Big) =  \operatorname{exp}\Big( \int_{t_{i}}^{t_{i+1}} \mathrm du \int_{u}^{t_{i+1}} \mathrm dv \, q(v-u) \Big)
		\end{equation*}
		for all $i\in \{0, \hdots, n-1\}$ and hence
		\begin{equation}
			\label{Equation: product of partition functions as probability under PPP}
			\prod_{i=0}^{n-1} Z_{t_{i+1}-t_i} = \operatorname{exp}\Big(\sum_{i=0}^{n-1}\int_{t_{i-1}}^{t_{i}} \mathrm du \int_{u}^{t_{i}} \mathrm dv \, q(v-u) \Big) = Z_t \operatorname{exp}\Big(-\int_{M} \mathrm du \mathrm dv \, q(v-u) \Big),
		\end{equation}
		with the set $M$ being given by
		\begin{equation*}
			M \coloneqq \big\{(u, v) \in [0, t]^2:\, u\leq t_i< v \text{ for some }i\in \{1, \hdots, n-1\} \big\}.
		\end{equation*}
		Let $\xi_{0, t} ,\xi_{t, \infty},\xi_{0, t, \infty}$ denote the restriction of $\xi$ to $[0, t]^2$, $[t, \infty)^2$ and $[0, t] \times (t, \infty)^2$ respectively. Then $\xi_{0, t}$ is a Poisson point process with intensity measure 
		\begin{equation*}
			q(v-u) \1_{\{0 < u < v < t\}} \, \mathrm du \mathrm dv
		\end{equation*}
		and we can hence restate \cref{Equation: product of partition functions as probability under PPP} as
		\begin{equation*}
			\frac{1}{Z_t}\prod_{i=0}^{n-1} Z_{t_{i+1}-t_i} = \mathbb P(\xi_{0,t}(M) = 0).
		\end{equation*}
		Now $\xi_{0, t} ,\xi_{t, \infty},\xi_{0, t, \infty}$ are independent. Since $\xi(M) = \xi_{0, t}(M)  + \xi_{0, t, \infty}(M)$, we obtain
		\begin{equation*}
			\frac{1}{Z_t}\prod_{i=0}^{n-1}Z_{t_{i+1}-t_i}  =  \mathbb P(\xi(M) = 0|\xi_{0, t, \infty}(\R^2) = 0) =  \mathbf P(X_{t_1} = \hdots = X_{t_{n-1}} = 0|X_t = 0).
		\end{equation*}
		In order to see that our process regenerates at time $T_1$ and that the first dormant and active period are independent, one can interpret $\xi$ in terms of queueing theory. We write our intensity measure as
		\begin{equation*}
			\beta \hat q(v-u) \1_{\{0 < u < v\}} \, \mathrm du \mathrm dv \quad  \text{ where } \beta = -E, \quad \hat q(t) \coloneqq \beta^{-1} q(t) \text{ for all }t>0.
		\end{equation*}
		Notice that $\hat q$ is a probability density as a consequence of Equation \cref{Equation: Arrival density is -E}.
		If we identify an atom $(s, t)$ of $\xi$ with a customer arriving at time $s$ and departing at time $t$, then $\xi$ is the law of a $M/G/\infty$ queue with arrival intensity $\beta$ and service time distribution with density $\hat q$ (this can be shown by using the marking theorem \cite[Theorem 5.6]{LastPenrose.2017} and the mapping theorem \cite[Theorem 5.1]{LastPenrose.2017} for Poisson point processes). That is, customers arrive according to a Poisson point process with intensity $\beta$ (i.e.\ the inter-arrival times are independent $\operatorname{Exp}(\beta)$ distributed) and depart after iid service times (which are independent of the arrival process) whose distribution has density $\hat q$. In particular, we have $d_1 \sim \operatorname{Exp}(\beta)$. Under this identification, the dormant and active periods are the idle and busy periods of the queue.
		
		It is left to show uniqueness, i.e.\ that there exists only one probability measure $\mathbf P$ on $\mathcal D$ satisfying the given assumptions such that
		\begin{equation*}
			\forall t\geq 0:\, e^{Et}Z_t = \mathbf P(N_t = 0).
		\end{equation*}
		We first notice that we can recover $\beta$ from differentiation of
		\begin{equation*}
			e^{Et}Z_t = \mathbf P(N_t = 0) = \mathbf P(d_1 \geq t) +  \mathbf P(d_1 < t, N_t = 0) = e^{-\beta t} + o(t)
		\end{equation*}
		in $t=0$, where we used in the last equality that
		\begin{equation*}
			\mathbf P(d_1 < t, N_t = 0) \leq \mathbf P(d_1<  t, a_1 \leq t-d_1) \leq (1-e^{-\beta t}) \mathbf P(a_1 < t) = o(t)
		\end{equation*}
		by independence of $d_1$ and $a_1$. As we will see in the upcoming proof of \cref{Proposition: Stieltjes transform in terms of renewal transform}, if we know the distribution of $d_1$, the renewal property allows us express the Laplace transform of $T_1$ in terms of the Laplace transforms of $d_1$ and $Z$. Hence, also the distribution of $T_1$ is uniquely determined. Hence, by independence of $d_1$ and $a_1$, the joint distribution of $(d_1, a_1)$ is uniquely determined, which then, again by the renewal property, determines the distribution $\mathbf P$ of the full process.
	\end{proof}
	
	\begin{proof}[Proof of \cref{Proposition: Stieltjes transform in terms of renewal transform}]
		We have for all $t\geq 0$
		\begin{align*}
			e^{Et} Z_t = \mathbf P(X_t = 0) &= \mathbf P(T_1 \leq t, X_{t} = 0) + \mathbf P(d_1 > t)\\
			&= \int_{\mathcal D} \1_{\{T_1(\omega) \leq t\}} \mathbf P(X_{t-T_1(\omega)} = 0) \, \mathbf P(\mathrm d\omega) + e^{-(m-E)t}.
		\end{align*}
		Multiplying by $e^{-Et}$ yields the renewal equation
		\begin{equation}
			\label{Equation: renewal equation for Z}
			Z_t = \mathbf E\big[\1_{\{T_1\leq t\}} e^{-E T_1} Z_{t-T_1}\big] + e^{-mt}
		\end{equation}
		for $Z$. Taking the Laplace transform in \cref{Equation: renewal equation for Z} and using the convolution property of the Laplace transform yields that, for all $z\in \mathbb C$ with $\operatorname{Re}(z)<E$,
		\begin{equation}
			\label{Equation: Laplace transformation of renewal equation}
			\mathcal L(Z)(-z) = \mathbf E\big[e^{(z-E) T_1}\1_{\{T_1 < \infty\}}\big]\mathcal L(Z)(-z) + \frac{1}{m-z}.
		\end{equation}
		Solving \cref{Equation: Laplace transformation of renewal equation} for 
		\begin{equation*}
			\mathcal L(Z)(-z) = \int_{[E, \infty)} \frac{1}{x-z} \mu(\mathrm dx)
		\end{equation*}
		now yields the claim.
	\end{proof}	
	\begin{proof}[Proof of \cref{Equation: Convergence of dormant probability to atom}]
		We will again use that $t \mapsto \mathbf P(X_t = 0)$ satisfies the renewal equation
		\begin{align}
			\label{Equation: Renewal equation for probability that we are dormant}
			 \forall t\geq 0:\, \mathbf P(X_t = 0) = \int_{\mathcal D} \1_{\{T_1(\omega) \leq t\}} \mathbf P(X_{t-T_1(\omega)} = 0) \, \mathbf P(\mathrm d\omega) + \mathbf P(d_1 > t).
		\end{align}
		Let us first assume that $\mathbf P(T_1 < \infty) = 1$.
		Since by assumption $\mu \neq \delta_E$, i.e.,\ $m>E$,
		the function $t\mapsto \mathbf P(d_1 > t) = e^{-(m-E)t}$ is decreasing and Lebesgue integrable and hence directly Riemann integrable \cite[Ch.~IV, Proposition 4.1]{Asmussen.2003}. By independence of $a_1$ and $d_1\sim \operatorname{Exp}(m-E)$, the distribution of $T_1$ under $\mathbf P$ is absolutely continuous with respect to the Lebesgue measure. Hence, \cref{Equation: Renewal equation for probability that we are dormant} combined with the key renewal theorem \cite[Ch.~IV, Theorem 4.3]{Asmussen.2003} yields the well-known formula
		\begin{equation}
			\label{Equation: limit of probability that we are dormant}
			\lim_{t\to \infty}\mathbf P(X_t = 0) = \frac{1}{\mathbf E[T_1]}\int_0^\infty \mathbf P(d_1 > t) \, \mathrm dt = \frac{\mathbf E[d_1]}{\mathbf E[T_1]},
		\end{equation}
		where the second and third expression in \cref{Equation: limit of probability that we are dormant} are by definition zero in case that $\mathbf E[T_1] = \infty$.
		
		For the case $\mathbf P(T_1 = \infty)>0$, one obtains from \cref{Equation: Renewal equation for probability that we are dormant} that
		\begin{equation*}
			\lim_{t\to \infty}\mathbf P(X_t = 0) = \frac{\lim_{t\to \infty} \mathbf P(d_1 > t)}{1-\mathbf P(T_1 < \infty)} = 0
		\end{equation*}
		see \cite[Ch.~VI, Proposition 5.4]{Asmussen.2003}. It is left to show that
		\begin{equation*}
			\lim_{t\to \infty} \mathbf E[D_t/t] =  \frac{\mathbf E[d_1]}{\mathbf E[T_1]}.
		\end{equation*}
		If $\mathbf E[T_1] < \infty$, this follows directly from \cite[Ch.~V, Theorem 3.1]{Asmussen.2003}.
		If $\mathbf P(T_1 = \infty)>0$, then the process is eventually active, i.e.,\ $t\mapsto D_t$ becomes eventually constant yielding
		 \begin{equation}
		 \label{Equation: convergence of averaged dormant time}
		 	\lim_{t\to \infty} \mathbf E[D_t/t] =  0,
		 \end{equation}
		by the dominated convergence theorem.
		Hence, it is left to show that \cref{Equation: convergence of averaged dormant time} holds if $\mathbf P(T_1 = \infty)=0$, but $\mathbf E[T_1] = \infty$. Let $N_t$ denote the number of renewal points (i.e.\ the number of end points of active periods) that lie in $[0, t]$, and let $d_k$ denote for $k\in \N$ the $k$-th dormant period. By the elementary renewal theorem \cite[Ch.~IV, Proposition 1.4]{Asmussen.2003} we have $N_t/t \to 0$ almost surely as $t\to \infty$. Since $N_t \to \infty$ almost surely as $t\to \infty$, we obtain with the law of large numbers
		\begin{equation*}
			 0\leq \limsup_{t \to \infty} D_t/t \leq \lim_{t\to \infty} \frac{N_t}{t}\frac{1}{N_t}\sum_{k=1}^{N_t} d_k = 0
		\end{equation*}
		almost surely. After taking the expected value, the dominated convergence theorem yields the claim.
	\end{proof}

	\begin{proof}[Proof of \cref{Theorem: Renewal representation Theorem 1 and Theorem 2}]
		We may again assume w.l.o.g. that $\mu$ is not a Dirac measure, i.e.,\ that $m>E$.
		Since we assume that $\mu(\{E\})>0$, we then have in particular $\mathbf E[T_1] < \infty$ by \cref{Corollary: classify singularity in E}. By the previous considerations it is left to show that
		\begin{equation*}
			2\int_{(E, \infty)}\frac{\mu(\mathrm dx)}{x-E} = \frac{(m-E)\mathbf E[a_1^2]}{\big(1+ (m-E)\mathbf E[a_1]\big)^2} =  \lim_{t\to \infty} \frac{\mathbf V[D_t]}{\mathbf E[D_t]}
		\end{equation*}
		and that the integral on the left hand side is finite if and only if $\mathbf E[a_1^2]< \infty$ (i.e.,\ if and only if $\mathbf E[T_1^2]< \infty$).
		We write with the monotone convergence theorem
		\begin{equation}
			\label{Equation: Substract atomic part from ST}
			\int_{(E, \infty)}\frac{\mu(\mathrm dx)}{x-E} = \lim_{\lambda \uparrow E} \int_{(E, \infty)}\frac{\mu(\mathrm dx)}{x-\lambda} =   \lim_{\lambda \uparrow E}\int_{[E, \infty)}\frac{\mu(\mathrm dx)}{x-\lambda} \,  - \, \frac{\mu(\{E\})}{E-\lambda}.
		\end{equation}
		To simplify notation, we define $\phi:(-\infty, E] \to \R$ by $\phi(\lambda) \coloneqq \mathbf E[e^{(\lambda - E)T_1}]$.
		By \cref{Proposition: Stieltjes transform in terms of renewal transform,Equation: Convergence of dormant probability to atom}, we have for every $\lambda < E$
		\begin{align}
			\label{Equation: Asymptotic exanspion of Stieljes transform close to lambda}
			\int_{[E, \infty)}\frac{\mu(\mathrm dx)}{x-\lambda} - \frac{\mu(\{E\})}{E-\lambda} &= \frac{1}{m - \lambda} \cdot \frac{1}{1 -\phi(\lambda)} -  \frac{1}{(E-\lambda)(m-E) \mathbf E[T_1]} \nonumber \\
			&=\frac{(E-\lambda)(m-E) \mathbf E[T_1] - (m-\lambda)(1 -\phi(\lambda))}{(m-\lambda)(E-\lambda)(m-E)\mathbf E[T_1] (1 -\phi(\lambda))} \nonumber  \\
			&=\frac{(m-E) \Big((E-\lambda) \mathbf E[T_1] -(1 -\phi(\lambda)\Big)  - (E-\lambda)(1 -\phi(\lambda))}{(m-\lambda)(E-\lambda)(m-E)\mathbf E[T_1](1 -\phi(\lambda)) } \nonumber  \\
			&=\frac{\frac{m-E}{E-\lambda} \big(\mathbf E[T_1] - f(\lambda) \big)- f(\lambda)}{(m-\lambda)(m-E) \mathbf E[T_1] f(\lambda)  },
		\end{align}
		where
		\begin{equation*}
			f(\lambda) \coloneqq \frac{1 -\phi(\lambda) }{E-\lambda}.
		\end{equation*}
		We have
		\begin{equation*}
			\lim_{\lambda \uparrow E} f(\lambda) = \phi'(E) = \mathbf E[T_1].
		\end{equation*}
		Let us assume that $\int_{(E, \infty)}(x-E)^{-1}\mu(\mathrm dx) < \infty$. Then \cref{Equation: Substract atomic part from ST,Equation: Asymptotic exanspion of Stieljes transform close to lambda} imply that the limit
		\begin{equation*}
			\lim_{\lambda \uparrow E} \frac{\mathbf E[T_1] - f(\lambda)}{E-\lambda} = \lim_{\lambda \uparrow E}\frac{1}{E-\lambda}  \Big(\phi'(E) - \frac{\phi'(E) - \phi(\lambda)}{E-\lambda} \Big)
		\end{equation*}
		exists and is finite. Since $\phi''$ is increasing, we have for all $\lambda < E$
		\begin{equation*}
			\phi''(\lambda) \leq \frac{2}{(E-\lambda)^2}\int_\lambda^E \mathrm ds \int_{s}^E \mathrm dt \, \phi''(t) =  \frac{2}{E-\lambda}  \Big(\phi'(E) - \frac{\phi'(E) - \phi(\lambda)}{E-\lambda} \Big).
		\end{equation*}
		Hence $\phi''(E) = \mathbf E[T_1^2] < \infty$ exists which in turn yields for $\lambda \leq E$
		\begin{equation*}
			\phi(\lambda) = 1 + (\lambda - E)\mathbf E[T_1] + \frac{1}{2}(\lambda - E)^2\mathbf E[T_1^2] + o((\lambda-E)^2),
		\end{equation*}
		which  with \cref{Equation: Asymptotic exanspion of Stieljes transform close to lambda} gives us
		\begin{equation*}
			\int_{[E, \infty)}\frac{\mu(\mathrm dx)}{x-\lambda} - \frac{\mu(\{0\})}{E-\lambda}  = \frac{ \tfrac{1}{2}(m-E)\mathbf E[T_1^2] - \mathbf E[T_1] + o(\lambda - E)}{(m-\lambda)(m-E)\mathbf E[T_1]^2 + o(\lambda - E)}.
		\end{equation*}
		The same holds, if we directly assume that $\mathbf E[T_1^2]< \infty$. By taking the limit $\lambda \uparrow E$, we hence obtain
		\begin{equation*}
			\int_{(E, \infty)}\frac{\mu(\mathrm dx)}{x-E} = \frac{\tfrac{1}{2}(m-E)\mathbf E[T_1^2] - \mathbf E[T_1]}{(m-E)^2 \mathbf E[T_1]^2 },
		\end{equation*}
		where the left hand side is finite if and only if the right hand side is finite, i.e.,\ if and only if $\mathbf E[T_1^2]< \infty$. Using independence of $d_1\sim \operatorname{Exp}(m-E)$ and $a_1$, we can further rewrite
		\begin{align*}
			&2\int_{(E, \infty)}\frac{\mu(\mathrm dx)}{x-E}\\
			&= \frac{(m-E)\mathbf V[T_1^2] + (m-E)\mathbf E[T_1]^2 - 2\mathbf E[T_1]}{(m-E)^2 \mathbf E[T_1]^2} \\
			&=\frac{(m-E)((m-E)^{-2} + \mathbf V[a_1^2]) + (m-E)((m-E)^{-1} + \mathbf E[a_1])^2 -2(m-E)^{-1}-2\mathbf E[a_1]}{(m-E)^2 \mathbf E[T_1]^2} \\
			&=\frac{(m-E)\mathbf V[a_1] + (m-E)\mathbf E[a_1]^2}{(m-E)^2\big((m-E)^{-1}+ \mathbf E[a_1]\big)^2} = \frac{(m-E)\mathbf E[a_1^2]}{\big(1+ (m-E)\mathbf E[a_1]\big)^2}.
		\end{align*}
		It is left to show that this expression coincides with the limit $\lim_{t\to \infty} \mathbf V[D_t]/\mathbf E[D_t]$ provided that $\mathbf E[a_1^2] <\infty$. With the known asymptotic variance of renewal-reward processes (see \cite{BrownSolomon.1975}, also compare to \cite[Ch. V, Theorem 3.2]{Asmussen.2003}), we have
		\begin{align*}
			\lim_{t\to \infty}\frac{1}{t} \mathbf V[D_t] &= \frac{1}{\mathbf E[T_1]} \Big(\mathbf V[d_1] + \frac{\mathbf E[d_1]^2}{\mathbf E[T_1]^2} \mathbf V[T_1] - 2 \frac{\mathbf E[d_1]}{\mathbf E[T_1]} \mathbf{Cov}(d_1, T_1) \Big) \\
			&= \frac{1}{\mathbf E[T_1]} \Big( \Big(1-\frac{\mathbf E[d_1]}{\mathbf E[T_1]}\Big)^2 \mathbf V[d_1] + \frac{\mathbf E[d_1]^2}{\mathbf E[T_1]^2} \mathbf V[a_1] \Big) \\
			&= \frac{1}{\mathbf E[T_1]} \Big(\frac{\mathbf E[a_1]^2}{\mathbf E[T_1]^2} \mathbf V[d_1] + \frac{\mathbf E[d_1]^2}{\mathbf E[T_1]^2} \mathbf V[a_1] \Big),
		\end{align*}
		where we used in the second equality that
		\begin{equation*}
			\mathbf{Cov}(d_1, T_1) = \mathbf V[d_1], \quad \mathbf V[T_1] = \mathbf V[d_1] + \mathbf V[a_1],
		\end{equation*}
		by independence of $d_1$ and $a_1$. Using $d_1\sim \operatorname{Exp}(m-E)$, we obtain
		\begin{equation*}
			\lim_{t\to \infty}\frac{1}{t} \mathbf V[D_t] = \frac{\mathbf E[a_1]^2 + \mathbf V[a_1]}{\mathbf E[T_1]^3(m-E)^2} = \frac{1}{(m-E)\mathbf E[T_1] }\frac{(m-E)\mathbf E[a_1^2]}{\big(1+(m-E)\mathbf E[a_1]\big)^2}
		\end{equation*}
		and the claim follows by dividing both sides by
		\begin{equation*}
			\frac{1}{(m-E)\mathbf E[T_1] } = \lim_{t\to \infty}\frac{1}{t} \mathbf E[D_t]. \qedhere
		\end{equation*}
	\end{proof}
	
	\section{Application to Generalized Spin-Boson Models}
	\label{Section: Generalized Spin-Boson models}
	\label{Section: Application to Spin systems}
	
	In this section, we now study the interaction of a finite dimensional quantum system, e.g., a spin system, with a bosonic quantum field.
	We will first present the considered model and state the obtained results together with a brief comparison to the literature.
	Proofs of the results, especially the application of \cref{Theorem: Convergence of averaged partition function,Theorem: Second order term}, are then presented in the following subsections.
	
	Let $\mathcal H$ be a finite dimensional Hilbert space.
	Adapting the terminology from \cite{BravyidiVincenzoOliveiraTerhal.2008},
	we will call a self-adjoint operator $A$ on $\mathcal H$ {\em stoquastic}
	with respect to an orthonormal basis $\mathcal B = \{\ph_i|i=1,\ldots,\dim\cH\}$ if 
	\begin{equation*}
		\langle \varphi_i, A \varphi_j \rangle \leq 0 \text{ for all }i\neq j.
	\end{equation*}
	To model the bosonic field with a given (arbitrary) bosonic Hilbert space $\fh$,
	we define the bosonic Fock space
	\begin{align*}
		\FS(\fh) \coloneqq \IC \oplus \bigoplus_{n=1}^\infty \fh^{\otimes_\sfs n},
	\end{align*}
	where the symbol $\otimes_\sfs$ denotes the {\em symmetric} tensor product.
	We further define the {\em second quantization} of a selfadjoint operator $S$ on $\fh$ as the selfadjoint operator
	\begin{align*}
		\dG(S) \coloneqq 0 \oplus \bigoplus_{n=1}^\infty \Big(\sum_{i=0}^{n-1} \Id^{\otimes i}\otimes S\otimes \Id^{\otimes(n-1-i)}\Big)^{**}
	\end{align*}
	and the {\em annihilation operator} for a given $f\in\fh$ by linear extension and closure of
	\begin{align*}
		a(f)\big(g_1\otimes_\sfs\cdots\otimes_\sfs g_n\big) \coloneqq \frac1{\sqrt n}\sum_{i=1}^n\braket{f,g_i}\big(g_1\otimes_\sfs\cdots\otimes\xcancel{g_i}\otimes_\sfs\cdots \otimes_\sfs g_n\big).
	\end{align*}
	The densely defined operator $a(f)$ and its adjoint satisfy the {\em canonical commutation relations}
	\begin{align*}
		[a(f),a(g)]=[a(f)^*,a(g)^*]=0, \quad [a(f),a(g)^*] = \braket{f,g}
	\end{align*}
	on a dense subspace of $\FS(\fh)$. If $f\in\sD(\varpi^{-1/2})$ for $\varpi$ being a selfadjoint strictly positive operator on $\fh$,
	i.e.\,$\braket{f,\varpi f}>0$ for all $f\in\sD(\varpi)\setminus\{0\}$ which especially implies injectivity of $\varpi$,
	then the relative bound
	\begin{align}
		\label{eq:a-relbound}
		\|a(f)\psi\| \le \|\varpi^{-1/2}f\|\|\dG(\varpi)^{1/2}\psi\|
	\end{align}
	holds for all $\psi\in\sD(a(f))\subset \sD(\dG(\varpi)^{1/2})$.
	For more details on Fock space calculus, we refer to the textbooks \cite{Parthasarathy.1992,Arai.2018}.
	
	We will from now on assume that $A$ and $B$ are self-adjoint operators on $\mathcal H$
	and that $B$ has an orthonormal basis $\mathcal B$ of eigenvectors such that $A$ is stoquastic with respect to $\mathcal B$.
	Furthermore, we assume $\varpi$ to be a selfadjoint strictly positive operator on $\fh$
	and choose $\nu\in\sD(\varpi^{-1/2})$.
	Under these assumptions, the generalized spin boson Hamiltonian \cite{AraiHirokawa.1997}
	\begin{equation}
		\label{Equation: Definition of generalized Spin-Boson Hamiltonian}
		H \coloneqq A \otimes \1 + 1 \otimes \mathrm d\Gamma(\varpi) + B \otimes (a(\nu) + a^*(\nu))
	\end{equation}
	is selfadjoint on its domain $\sD(H)=\mathcal H \otimes \sD(\dG(\varpi))$, by the relative bound \cref{eq:a-relbound}, the canonical commutation relations and the Kato--Rellich theorem.
	
	One might, for example, consider the case where $A$ is the Hamiltonian of a quantum spin system composed out of $n$ qubits, meaning $\mathcal H = (\mathbb C^2)^{\otimes n}$,
	which we couple to a bosonic field via \cref{Equation: Definition of generalized Spin-Boson Hamiltonian} choosing $B = \sum_{i=1}^n \alpha_i \sigma_z^i$ for some constants $\alpha_1, \hdots, \alpha_n \in \R$.
	Then any $A$ which is stoquastic with respect to the usual $z$-basis given by
	\begin{equation*}
		\mathcal B \coloneqq \big\{\ket{z_1} \otimes \hdots \hdots \otimes\ket{z_n}:\, z\in \{-1, 1\}^n\big\} \text{ where } \ket{1} = \begin{pmatrix}
			1 \\ 0
		\end{pmatrix},\,
		\ket{-1} = \begin{pmatrix}
			0 \\ 1
		\end{pmatrix}
	\end{equation*}
	satisfies our standing assumptions.
	To give two concrete examples, the Hamiltonian of the ferromagnetic Heisenberg-model on a finite graph  $G = ([n], E)$ given by 
	\begin{equation*}
		A = -\sum_{\{i, j\} \in E  } \sigma_x^i \sigma_x^j +  \sigma_y^i \sigma_y^j + \sigma_z^i \sigma_z^j
	\end{equation*}
	is stoquastic in that sense,
	and so is the standard spin-boson (SSB) model for which $n=1$ and $A = -\sigma_x$. We refer the reader to \cite{BravyidiVincenzoOliveiraTerhal.2008} for other examples. 
	
	We call a unit vector $\phi \in \sD(H)$ a {\em ground state} of $H$ if it is a eigenvector of $H$ to the eigenvalue 
	\begin{equation*}
		E \coloneqq \inf \sigma(H).
	\end{equation*}
	In the following, we will show that our Wiener-type theorem directly generalizes and partially strengthens the probabilistic criteria for the existence and non-existence of ground states for the SSB model, which were given in \cite{HaslerHinrichsSiebert.2021a,BetzHinrichsKraftPolzer.2025}, to the general setup \cref{Equation: Definition of generalized Spin-Boson Hamiltonian}. Since our existence criterion has a natural interpretation in functional analytic terms, we will state the latter first.
	
	We will study the spectral measure $\mu$ of $H$ with respect to the unit vector 
	\begin{equation}
		\label{Equation: Definition of psi}
		\psi = \frac{1}{\sqrt{\dim\cH}}\sum_{k=1}^{\dim\cH} \varphi_k \otimes \Omega.
	\end{equation}
	The assumption that $A$ is stoquastic yields
	\begin{align}
		\label{eq:GSenergy}
		E = \inf \supp \mu = \inf \sigma(H),
	\end{align}
	see \cref{subsec:positivity} for a proof.
	Let $P_E = \1_{\{E\}}(H)$ denote the orthogonal projection onto $\ker(H-E)$ such that in particular $\mu(\{E\})  = \langle \psi, P_E \psi \rangle$.
	We then have
	\begin{align}
		\label{eq:GSequivalence}
		\rho \coloneqq \mu(\{E\}) > 0 \iff \text{$H$ has a ground state}
	\end{align}
	and in this case
	\begin{equation}
		\label{Equation: Projection of psi onto eigenspace}
		\phi \coloneqq P_E\psi/\|P_E\psi\| \in \ker(H-E).
	\end{equation}
	We note that the implication ``$\Rightarrow$'' in \cref{eq:GSequivalence} directly follows from \cref{eq:GSenergy},
	whereas we will prove the reverse implication in \cref{subsec:positivity}.
	To derive (non-)existence criteria, we will thus derive upper and lower bounds on $\rho$.
	To state the later in a cleaner fashion, we will state them in terms of $\operatorname{log}(1/\rho)$ where $\log(1/0) \coloneqq \infty$.
	
	Let us now first consider the case of massive bosons, i.e., $\inf\sigma(\varpi) > 0$, for which it is well-known that $\rho>0$ \cite{AraiHirokawa.1997}, a fact which we will also reprove in \cref{Corollary: Existence in infra-red regular case}. We denote by $\mathbf N = \Id_\cH \otimes \dG(\Id_\fh)$ the boson number operator.
	\begin{thm}
		\label{Theorem: estimate vacuum overlap generalized Spin-Boson} 
		Assume that $\inf\sigma(\varpi) > 0$ and let $\phi$ be defined as in \cref{Equation: Projection of psi onto eigenspace}. Then
		\begin{equation*}
			\log(1/\rho) \leq \log(\dim\cH) + \langle \phi,\, \mathbf N \phi\rangle .
		\end{equation*}
	\end{thm}
	\begin{proof}
		This will follow directly from \cref{Theorem: Estimates Spin-Boson models,Proposition: Representation of upper bound as expected boson number} below.
	\end{proof}
	This estimate in particular allows us to state a criterion for the existence of ground states by introducing an infrared regularization of $H$,
	a procedure often used in functional analytic proofs for the existence of ground states as well, cf. \cite{Gerard.2000,GriesemerLiebLoss.2001,HaslerHinrichsSiebert.2021a} and references therein.
	A possible regularization procedure is to replace $\varpi$ in \cref{Equation: Definition of generalized Spin-Boson Hamiltonian} by $\varpi + \eps\Id_\fh$, i.e.,
	defining $H_\varepsilon \coloneqq H + \varepsilon\mathbf N$.
	Then $H_\varepsilon$ has a ground state $\phi_\varepsilon$ (defined analogously to \cref{Equation: Projection of psi onto eigenspace}) for every $\eps>0$  and, combined with upper semicontinuity of $\rho$ in $\varepsilon$ which we prove in \cref{Lemma: continuity vacuum overlap}, we obtain
	\begin{cor}
		\label{Corollary: Existence criterion generalized Spin-Boson}
		Assume that $\liminf_{\varepsilon \downarrow 0} \langle \phi_\varepsilon, \mathbf N \phi_\varepsilon\rangle < \infty$. Then $H$ has a ground state.
	\end{cor}	
	\begin{proof}
		This will follow directly from \cref{Theorem: estimate vacuum overlap generalized Spin-Boson,Lemma: continuity vacuum overlap}.
	\end{proof}
	For the proof of the essential \cref{Theorem: Estimates Spin-Boson models},
	we will apply \cref{Theorem: Convergence of averaged partition function} to the spectral measure $\mu$ by using the Feynman--Kac representation of its Laplace transform.
	To state the Feynman--Kac formula, let us recall that $\mathcal B = \{\varphi_1, \hdots, \varphi_{\dim \cH}\}$ is an eigenbasis of $B$ such that $A$ is stoquastic with respect to $\mathcal B$.
	We define $v, w: \{1, \hdots, \dim\cH\}\to \mathbb R$ by 
	\begin{equation*}
		v(i) \coloneqq -\sum_{j=1}^{\dim\cH} \langle \varphi_i, A \varphi_j \rangle, \quad w(i) \coloneq  -\langle \varphi_i, B \varphi_i \rangle.
	\end{equation*}
	The stoquasticity of $A$ implies that $-A$ differs, in the basis $\mathcal B$, only by a diagonal matrix from the generator of a Markov process on $\{1, \hdots, \dim\cH\}$. That is, if we set
	\begin{equation*}
		Q = -A - \sum_{i=1}^{\dim\cH} v(i) \langle \varphi_i, \, \cdot \, \rangle \varphi_i 
	\end{equation*}
	then
	\begin{equation*}
		Q_{ij} \coloneqq \langle \varphi_i, Q \varphi_j \rangle \geq 0 \quad \text{for all }i\neq j, \quad Q_{i i} = -\sum_{j\neq i} Q_{i j}\quad  \text{for all }i.
	\end{equation*}
	We denote by $X$ the Markov-process generated by $Q$, started in the uniform distribution on $\{1, \hdots, \dim\cH\}$, i.e., for all $i,j=1,\ldots,\dim\cH$ and $t\ge s\ge0$
	\begin{equation}
		\label{eq:Markovgen}
		\mathbb P(X_0 = j) = 1/d, \quad \mathbb P(X_t = j|X_s = i) = \big\langle \varphi_i, e^{(t-s)Q} \varphi_j \big\rangle.
	\end{equation}
	We will denote expectations with respect to the probability measure $\mathbb P$ by $\mathbb E$.
	
	Finally, we define the function $g:\R \to (0, \infty)$ by
	\begin{equation}
		\label{def:g}
		g(t) \coloneq  \Braket{\nu,e^{-\lvert t\lvert\varpi}\nu}_\fh.
	\end{equation}
	We can now state the Feynman--Kac representation of the Laplace transform of $\mu$ (recall the definition \cref{def:Laplace}), for which we present a simple proof in \cref{subsec:FK}.
	\begin{prop}
		\label{Proposition: Feynman-Kac}
		For all $T\geq 0$
		\begin{equation*}
			Z_T = \langle \psi, e^{-TH} \psi\rangle = \mathbb E\Big[\exp\Big( \frac{1}{2}\int_{[0, T]^2} g(t-s) w(X_s)w(X_t) \, \mathrm ds \mathrm dt +  \int_{[0, T]} v(X_s) \, \mathrm ds \Big)   \Big].
		\end{equation*}
	\end{prop}
	We point out that for the special case of the SSB model, there are more general versions of the Feynman--Kac formula known, see \cite{HirokawaHiroshimaLorinczi.2014,HaslerHinrichsSiebert.2021c}.
	Combining \cref{Proposition: Feynman-Kac} with \cref{Theorem: Convergence of averaged partition function} allows us to study the existence of ground states, or more precisely the value of $1/\rho$, by studying $Z_T$. We will understand $Z_T$ as the partition function (i.e.\ the normalization constant) of the perturbed path measure $\widehat{\mathbb P}_{T}$ defined by
	\begin{equation}
		\label{Equation: Definition of path measure}
		\widehat{\mathbb P}_{T}(\mathrm dX) = \frac{1}{Z_T} \exp\Big( \frac{1}{2}\int_{[0, T]^2} g(t-s) w(X_s)w(X_t) \, \mathrm ds \mathrm dt +  \int_{[0, T]} v(X_s) \, \mathrm ds  \Big) \, \mathbb P(\mathrm dX).
	\end{equation}
	We define  $\widehat{\mathbb P}_{s, t} \coloneqq \widehat{\mathbb P}_{s} \otimes \widehat{\mathbb P}_{t}$ and denote by $\widehat{\mathbb E}_{s, t}$ the expected value taken with respect to $\widehat{\mathbb P}_{s, t}$. Furthermore, we denote by $(X, Y)$ a pair of random $\{1, \hdots, \dim\cH\}$ valued functions, either drawn from $\widehat{\mathbb P}_{s, t}$ or from the unbiased product measure $\mathbb P^{\otimes 2}$. Introducing $\widehat{\mathbb P}_{T}$ allows us to express fractions of partition functions (as they appear in \cref{Theorem: Convergence of averaged partition function,Theorem: Second order term}) in terms of the perturbed path measure, see \cref{Proposition: Fractions of partition functions in terms of path measure} later in the text. Applying this representation then yields
	\begin{thm}
		\label{Theorem: Estimates Spin-Boson models}
		We have
		\begin{align}
			\log(1/\rho) &\leq \log(\dim\cH) + \liminf_{T \to \infty} \frac{1}{2T}\int_{[0, T]^2} |t-s|g(t-s) \widehat{\mathbb E}_{T}[w(X_s) w(X_t)] \, \mathrm ds \mathrm dt,
			\label{eq:SBupperbound}
			\\
			\log(1/\rho) &\geq \limsup_{T \to \infty} \int_{[0, T]^2} g(t+s) \widehat{\mathbb E}_{T, T}\big[w(X_s) w(Y_t)|X_0 = Y_0\big] \, \mathrm ds \mathrm dt.
			\label{eq:SBlowerbound}
		\end{align}
	\end{thm}
	
	Hence, we can study the existence and non-existence of ground states by studying the decay of correlations of the stochastic process $(X_t)_{t\geq 0}$ under the perturbed measure $\widehat{\mathbb P}_{T}$ in the limit $T\to \infty$.
	
	Let us now state some simple observations that follow directly from \cref{Theorem: Estimates Spin-Boson models}.
	The first observation is the well-known existence of ground states for models with infrared-regular coupling \cite{BachFroehlichSigal.1998a,Gerard.2000,DamMoller.2018a}.
	\begin{cor}
		\label{Corollary: Existence in infra-red regular case}
		If $\nu\in \sD(\varpi^{-1})$, then $H$ has a ground state.
	\end{cor}
	\begin{rem}
		Especially, if $\inf\sigma(\varpi)>0$, then $\varpi^{-1}\in\cB(\fh)$, so $\sD(\varpi^{-1})=\fh$, i.e., $H$ has a ground state for arbitrary coupling functions $\nu\in\fh$.
	\end{rem}
	\begin{proof}
		First note that, by combining the spectral theorem for $\varpi$ and Fubini's theorem, for any $\alpha>0$, we have
		\begin{align}
			\label{eq:Fubini}
			\int_0^\infty s^{2\alpha-1} g(s)\,\rmd s = \Gamma(2\alpha)\|\varpi^{-\alpha}\nu\|_\fh^2,
		\end{align}
		where the right hand side is defined to be infinite if $\nu\notin\sD(\varpi^{-\alpha})$.
		
		The case $\alpha=1$ yields that our infrared-regularity assumption $\nu\in\sD(\varpi^{-1})$ is equivalent to $G=\int_0^\infty sg(s)\,\rmd s<\infty$.
		Setting $W=\max_{i=1,\ldots,\dim\cH}w(i)$, we can thus estimate the right hand side in \cref{eq:SBupperbound}
		by
		$\log(1/\rho) \le \log(\dim\cH) + W^2G$, which proves $\rho>0$ and thus the statement.
	\end{proof}
	We can significantly weaken this condition if the correlation functions decay significantly fast.
	The proof employs the upper bound \cref{eq:SBupperbound} and \cref{eq:Fubini} similar to the previous one, whence we omit details here.
	\begin{cor}
		\label{Corollary: Concrete estimates on necessary correlation decay}
		If $\nu\in\sD(\varpi^{-\alpha})$ for some $\alpha\ge\frac12$,
		then
		\begin{equation*}
			\liminf_{T\to\infty}\operatorname{sup}_{0\leq s, t \leq T} |t-s|^{2(1-\alpha)}\widehat{\mathbb E}_{T}[w(X_s) w(X_t)] < \infty
		\end{equation*}
		implies the existence of a ground state.
	\end{cor}
	We can also give a criterion for the absence of ground states in infrared-critical cases, in view of \cref{eq:GSequivalence}.
	\begin{cor}
		\label{Corollary: Long range order implies absence in infra-red critical model}
		If $\nu\notin\sD(\varpi^{-1})$, then
		\begin{equation*}
			\liminf_{T \to \infty}\inf_{ 0\leq s, t \leq T}\widehat{\mathbb E}_{T, T}\big[w(X_s) w(Y_t)|X_0 = Y_0\big] > 0
		\end{equation*}
		implies $\rho=0$ and thus the absence of ground states. In particular, if all eigenvalues of $B$ are non-zero and have the same sign, then $H$ does not have a ground state.
	\end{cor}
	\begin{proof}
		Inserting the assumption into \cref{eq:SBlowerbound} directly implies that there is $C>0$ such that
		\[\log(1/\rho)\ge C \lim_{T\to\infty}\int_{[0,T]^2}g(t+s)\rmd s\rmd t = C\int_0^\infty rg(r)\rmd r = \infty, \]
		where the last equality follows from the assumption by \cref{eq:Fubini} and thus proves the statement.
	\end{proof}
	\begin{rem}
		In \cite[Theorem 3.1]{AraiHirokawaHiroshima.1999} it was shown that for $\nu\notin\sD(\varpi^{-1})$ and strictly positive $B$ there does not exists a ground state that lies in $\sD(\mathbf N^{1/2})$. \cref{Corollary: Long range order implies absence in infra-red critical model} removes this restriction on the domain.
	\end{rem}
	
	For the specific case of the SSB model, where $H$ is given by \cref{Equation: Definition of generalized Spin-Boson Hamiltonian}
	with $\mathcal H = \mathbb C^2$, $A = -\sigma_x$ and $B=\alpha \sigma_z$, let us compare the previous results to the existing criteria for the (non-)existence of ground states given in \cite{HaslerHinrichsSiebert.2021a,HaslerHinrichsSiebert.2021c,BetzHinrichsKraftPolzer.2025}. We start by noticing that our stochastic process $X$ is a now a continuous time random walk on $\{-1, 1\}$ with $\operatorname{Exp}(1)$ distributed jumping times. By using the symmetry of the model, we will show that we can drop the factor $\log 2$ in \cref{Theorem: estimate vacuum overlap generalized Spin-Boson,Theorem: Estimates Spin-Boson models}.
	
	\begin{prop}
		\label{Proposition: Estimates in sepcial case of Spin-Boson}
		In the case of the SSB model
		\begin{equation}
			\label{Equation: Existence criterion Spin-Boson}
			\log(1/\rho) \leq \liminf_{T \to \infty} \frac{1}{2T}\int_{[0, T]^2} |t-s|g(t-s) \widehat{\mathbb E}_{T}[X_s X_t] \, \mathrm ds \mathrm dt.
		\end{equation}
		If additionally $\inf\sigma(\varpi)>0$ holds, then
		\begin{equation}
		\label{Equation: Existence criterion Spin-Boson in terms of number operator}
			\log(1/\rho) \leq \, \langle \phi, \mathbf N \phi\rangle.
		\end{equation}
	\end{prop}
	In particular, if the right hand side of \cref{Equation: Existence criterion Spin-Boson} is finite, then $H$ has a ground state.
	Note that by the estimate $e^{-|x|}\leq |x|^{-1}$ and the spectral theorem, we have $0\le g(t)\le t^{-1}\|\varpi^{-1/2}\nu\|_\fh$,
	whence our existence criterion strengthens (and in fact generalizes) the implication
	\begin{equation*}
		\limsup_{T \to \infty} \frac{1}{T}\int_{[0, T]^2} \widehat{\mathbb E}_T[X_{s} X_t] \, \mathrm ds \mathrm dt  <\infty  \implies  H \text{ has a ground state}
	\end{equation*}
	which was proven in \cite{HaslerHinrichsSiebert.2021a,HaslerHinrichsSiebert.2021c} and verified for small $\alpha$ in \cite{HaslerHinrichsSiebert.2021b}, also see \cite{BetzHinrichsKraftPolzer.2025} for a review of these results.
	In particular, for the physically important case $\HS=L^2(\IR^3)$, $\varpi f(k)=|k|f(k)$ and $v(k)=|k|^{-1/2}$, we have $g(t)\sim t^{-2}$ and hence arbitrary slow polynomial decay of correlations implies the existence of a ground state.
	
	Moreover, one easily checks that
	\begin{equation*}
		\widehat{\mathbb E}_{T, T}\big[X_s Y_t|X_0 = Y_0\big] = \widehat{\mathbb E}_{T}\big[X_s X_0]\cdot \widehat{\mathbb E}_{T}\big[X_t X_0].
	\end{equation*}
	Hence, the criterion for the absence of a ground state in the infra-red critical case as given by \cref{Corollary: Long range order implies absence in infra-red critical model} in the case of the SSB model coincides with the one proven in \cite[Cor.~3.5]{BetzHinrichsKraftPolzer.2025}.
	
	We finish this section by mentioning that the aforementioned results on the SSB model remain valid for polaron-type Hamiltonians at total momentum zero which are given by
	\begin{equation}
		\label{Definition: Polaron Hamiltonian at total momentum zero}	
		H = \frac{1}{2}P_\sff^2 + \dG(\varpi) + \frac{\sqrt{\alpha}}{\sqrt{2} \pi}  (a(\nu) +  a(\nu)^*)
	\end{equation}
	where we have chosen $\HS=L^2(\IR^d)$, $\varpi$ is a multiplication operator and $P_\sff = \mathrm d\Gamma(\operatorname{id}_{\mathbb R^d})$ denotes the momentum of the field.
	The Feynman--Kac formula \cite{Feynman.1955}, also see \cite{HinrichsMatte.2024} for a recent generalization, in this case states that
	\begin{equation*}
		\langle \Omega, e^{-TH} \Omega \rangle =  \mathbb E\Big[ \exp\Big( \int_{[0, T]^2} \omega(t-s, X_s - X_t) \, \mathrm ds \mathrm dt \Big)     \Big]
	\end{equation*}
	where $(X_t)_{t\geq 0}$ is a standard Brownian motion on $\mathbb R^3$ and where
	\begin{equation*}
		\omega(t-s, X_s - X_t) =  \alpha\int |\nu(k)|^2 e^{\mathrm i k \cdot(X_s - X_t) } e^{-\varpi(k) |t-s|} \, \mathrm dk.
	\end{equation*} 
	We set $\rho \coloneqq \langle \Omega, \Id_{\{E \}}(H)\Omega \rangle$ where $E = \inf \sigma(H)$. By the same arguments as in the proof of \cref{Proposition: Estimates in sepcial case of Spin-Boson} (instead of \cref{Equation: Spin product in terms of jumping process} one here uses the independence of Brownian increments) one obtains
	\begin{equation*}
		\log(1/\rho) \leq \limsup_{T \to \infty} \frac{1}{2T}\int_{[0, T]^2} |t-s|\widehat{\mathbb E}_{T}[\omega(t-s, ,X_s-X_t)] \, \mathrm ds \mathrm dt
	\end{equation*}
	as well as
	\begin{equation*}
		\log(1/\rho) \leq \, \langle \phi, \mathbf N \phi\rangle
	\end{equation*}
	provided that $\inf_k \varpi(k)>0$.
	Here, the expected value $\widehat{\mathbb E}_{T}$ is taken with respect to the accordingly defined perturbed path measure with partition function $Z_T = \langle \Omega, e^{-TH} \Omega \rangle$. 
	
	\subsection{Ground State Energy and Existence}
	\label{subsec:positivity}
	We here prove that the minimum of the support of $\mu$ is the ground state energy $E$ \cref{eq:GSenergy} and the fact that $\rho=0$ implies absence of ground states \cref{eq:GSequivalence}.
	
	The arguments presented here are somewhat standard in the literature and rely on Perron--Frobenius--Faris theory, see for example \cite[\textsection\,XIII.12]{ReedSimon.1978}.
	To apply it in the way formulated therein, we need to unitarily map our Hilbert space $\HS\otimes\FS(\fh)$ onto an $L^2$-space $L^2(\cM,\lambda_\cM)$ for some appropriately chosen measure space $(\cM,\fM,\lambda_\cM)$.
	In the case $\HS=\IC$, this is done using the Wiener--It\^{o}--Segal isomorphism \cite[\S I.3]{Simon.1974}. 
	\begin{prop}
		\label{prop:Theta}
		There exists a probability space $(\cQ,\fQ,\lambda_\cQ)$ and a unitary $\Theta:\FS(\fh)\to L^2(\cQ,\lambda_\cQ)$ such that
		\begin{enumprop}
			\item $\Theta\Omega=1$,
			\item $\Theta\big(a(\nu)+a(\nu)^*\big)\Theta^*$ is a multiplication operator,
			\item $\Theta e^{-t\dG(\varpi)}\Theta^*$ preserves positivity for any $t\ge 0$, i.e., $\Theta e^{-t\dG(\varpi)}\Theta^* f\ge 0$ almost everywhere if $f\ge 0$ almost everywhere,
		\end{enumprop}
	\end{prop}
	\begin{proof}
		The statements can be found in \cite[Thms.~I.11,\,I.12]{Simon.1974},
		under the additional assumption that there exists a complex conjugation $C$ on $\fh$ such that both $\varpi$ and $\nu$ are $C$-real.
		We prove in \cref{lem:existenceconjugation} below that such a complex conjugation always exists.
	\end{proof}
	\begin{lem}
		\label{lem:existenceconjugation}
		Given a selfadjoint operator $S$ on $\fh$ and a vector $\xi\in\fh$,
		there exists a complex conjugation $C$ on $\fh$, i.e., an anti-unitary involution, such that both $\xi$ and $S$ are $C$-real, i.e., $C\xi=\xi$ and $SC=CS$.
	\end{lem}
	\begin{proof}
		By the spectral theorem, there exists a measure space $(M,\Sigma,\sigma)$, a $\Sigma$-measurable function $f:M\to\IR$ and a unitary $U: L^2(M,\Sigma,\sigma)\to \fh$ such that
		$U^*SU=f$ as a multiplication operator. Now we define $C$ acting on $\zeta\in \fh$ as
		\[ (U^*CU\zeta)(x)  = \Big(\frac{U\xi(x)^2}{|U\xi(x)|^2}\1_{\{U\xi(x) \neq 0\}} +\1_{\{U\xi(x) = 0\}}\Big)\overline{U\zeta(x)}. \]
		The fact that $\xi$ and $S$ are $C$-real follows by direct calculation and the fact that $f$ is real-valued.
	\end{proof}
	We now fix $(\cQ,\fQ,\lambda_\cQ)$ as in \cref{prop:Theta}.
	We can then define $\cM = \{1,\ldots,\dim\HS\}\times \cQ$ and equip it with the product measure $\lambda_\cM=\delta\otimes\lambda_\cQ$, where $\delta$ denotes the counting measure.
	The desired unitary $U:\HS\otimes\HS(\fh)\to L^2(\cM,\lambda_\cM)$ is then uniquely determined by
	\begin{align*}
		\big(U(\ph\otimes\xi)\big)(i,q) = \braket{\ph_i,\ph} (\Theta \xi)(q).
	\end{align*}
	Especially, note that our test vector $\psi$ given in \cref{Equation: Definition of psi} is mapped to the constant function $U\psi = (\dim\HS)^{-1/2}$.
	Important for our main observations is the following statement.
	\begin{cor}
		\label{cor:pospres}
		The operator $Ue^{-tH}U^*$ is positivity preserving for any $t\ge 0$.
	\end{cor}
	\begin{proof}
		Note that the assumption that $A$ is stoquastic with respect to the basis $\cB$ immediately implies $Ue^{-tA\otimes\Id_{\FS(\fh)}}U^*$ preserves positivity.
		Furthermore, $Ue^{-t\Id_\HS\otimes \dG(\varpi)}U^*$ preserves positivity by the definition of the Wiener--It\^{o}--Segal isomorphism.
		Since further $UB\otimes \big(a(\nu)+a(\nu)^*\big)U^*$ is a multiplication operator by construction,
		$Ue^{-tB\otimes \big(a(\nu)+a(\nu)^*\big)}U^*$ is positivity preserving and thus the fact that $Ue^{-tH}U^*$ preserves positivity follows from the Trotter product formula.
	\end{proof}
	This already gives us the 
	\begin{proof}[Proof of \cref{eq:GSenergy}]
		Since $U\psi$ is a strictly positive test vector,
		the statement immediately follows from \cref{cor:pospres};
		also see \cref{rem:pospres} or \cite[Thm.~C.1]{MatteMoller.2018}.
	\end{proof}
	To prove that $\rho=0$ implies absence of ground states,
	we will use the following simple version of Perron--Frobenius--Faris theory \cite{Faris.1972}.
	\begin{prop}
		\label{prop:PFF}
		If $A$ is a positivity preserving self-adjoint operator on $L^2(\cM,\fM,\lambda)$
		and if $\|A\|$ is an eigenvalue of $A$, there exists a normalized non-negative eigenfunction.
	\end{prop}
	\begin{proof}
		Throughout this proof we use complex conjugation and positive/negative parts of functions defined pointwise.
		Let $f\in L^2(\cM,\fM,\lambda)$ satisfy $Af=\|A\|f$. Then since $A$ is positivity preserving and thus real w.r.t. to the usual pointwise complex conjugation, we have
		\[ A(f+\overline{f}) = Af + \overline{Af} = \|A\|(f+\overline{f}), \]
		i.e., $f+\overline{f}$ is an eigenfunction to the same eigenvalue. Hence, we can from now assume $f$ is a.e. real-valued.
		Then $f_+,f_-\ge 0$ a.e., so since $A$ preserves positivity we find
		\begin{align*} \|A\|\|f\| &=  \braket{f,Af} = \braket{f_+,Af_+} + \braket{f_-,Af_-}-2 \braket{f_-,Af_+} \\&\le \braket{f_+,Af_+} + \braket{f_-,Af_-} +2 \braket{f_-,Af_+} = \braket{|f|,A|f|}\le \|A\|\|f\|. \end{align*} 
		By spectral decomposition, $\braket{|f|,A|f|} = \|A\|\|f\|$ implies that $|f|$ is an eigenfunction of $A$ to the eigenvalue $\|A\|$, which proves the statement since $|f|$ is normalized whenever $f$ is.
	\end{proof}
	This now immediately yields the
	\begin{proof}[Proof of \cref{eq:GSequivalence}]
		The implication $\Rightarrow$ is obvious from \cref{eq:GSenergy}.
		To prove the reverse direction, let us fix some $t>0$. Note that $\ker(H-E)=\ker(e^{-tH}-e^{-tE})$ where $e^{-tE} = \|e^{-tH}\|$. Hence, if $\ker(H-E) \neq \{0\}$, there exists by \cref{cor:pospres,prop:PFF} a non-zero $\phi \in \ker(H-E)$ such that $U\phi \geq 0$. Then
		$\rho\ge \Braket{\phi,\psi}>0$ since $U\psi$ is strictly positive.
	\end{proof}
	\subsection{Feynman--Kac Formula (\cref{Proposition: Feynman-Kac})}
	\label{subsec:FK}
	Let us now prove the Feynman--Kac formula. Since proofs of the later are fairly well known for related models,
	see for example the textbook \cite{LorincziHiroshimaBetz.2011} or \cite{HirokawaHiroshimaLorinczi.2014} for the SSB model,
	we defer some technical details to the literature.
	
	In the following, exponentials of unbounded operators
	are defined by their series expansion with domain being all Hilbert space vectors, such that the series converges pointwise.
	\begin{lem}
		\label{Lemma: Product of Is}
		If $f,g\in \sD(\varpi^{-1/2})$ and $t>0$,
		then the operator $e^{a(f)^*}e^{-t\dG(\varpi)}e^{a(g)}$ has a unique bounded extension $I_t(f,g)$.
		Furthermore, given $\tilde f,\tilde g\in\sD(\varpi^{-1/2})$ and $\tilde t>0$, we have
		\begin{align*}
			I_{\tilde t}(\tilde f,\tilde g)I_t(f,g) = e^{\braket{\tilde g,f}}I_{\tilde t+t}(\tilde f + e^{-\tilde t\varpi}f,e^{-\tilde t\varpi}\tilde g + g).
		\end{align*}
	\end{lem}
	\begin{proof}
		The operator $e^{a(f)^*}e^{-t\dG(\varpi)}e^{a(g)}$ clearly contains the span $\fF$ of vectors of the form $h_1\otimes_\sfs\cdots\otimes_\sfs h_n$ in its domain
		and is thus densely defined. Boundedness furthermore follows from the estimate \cref{eq:a-relbound} (we refer the reader to \cite[Appendix~6]{GueneysuMatteMoller.2017} for details on how to estimate the series expansions) so that the existence of a unique bounded extension is proven.
		
		The second statement again follows on $\fF$ (and thus on all of $\FS(\fh)$),
		since the canonical commutation relations combined with the Baker--Campbell--Hausdorff Formula imply
		\[e^{a(\tilde g)}e^{a(f)^*} = e^{\frac12 \braket{\tilde g,f}}e^{a(\tilde g)+a(f)^*} = e^{\braket{\tilde g,f}}e^{a(f)^*}e^{a(\tilde g)}\]
		and since $a(h)e^{-t\dG(\varpi)} = e^{-t\dG(\varpi)}a(e^{-t\dG(\varpi)}h)$ holds on $\fF$ for any $h\in \fh$, whence
		\[ e^{a(\tilde g)}e^{-t\dG(\varpi)} = e^{-t\dG(\varpi)}e^{a(e^{-t\varpi}\tilde g)},
		\quad e^{-t\dG(\varpi)}e^{a(\tilde f)^*} = e^{a(e^{-t\varpi}\tilde f)^*}e^{-t\dG(\varpi)}. \qedhere  \]
	\end{proof}
	\begin{proof}[Proof of \cref{Proposition: Feynman-Kac}]
		We set $\phi \coloneqq (\dim \cH )^{-1/2} \sum_i \varphi_i$ such that $\psi = \phi \otimes \Omega$.
		We denote by $P_i = \braket{\ph_i,\cdot}\ph_i$ the orthogonal projection on the span of $\ph_i$
		and define the operator-valued function
		\begin{align*}
			F(t) = \sum_{i=1}^{\dim\cH}\big(e^{-tA} P_i\big)\otimes I_t\big(tw(i)\nu,tw(i)\nu\big).
		\end{align*}
		One can readily check, again see \cite[Appendix~6]{GueneysuMatteMoller.2017} for details, that
		$[0,\infty)\to\cB(\HS\otimes\FS(\hs))$ is strongly continuous and that the strong right-derivative exists and satisfies $\partial_t^+ F(t)|_{t=0} = -H$. 
		Thus, by applying the Chernoff product formula \cite{Chernoff.1968}, we then find
		\begin{align*}
			\braket{\psi,e^{-tH}\psi}
			&= \lim_{n\to\infty}\Braket{\psi,\big(F(t/n)\big)^n\psi}\\
			&=  \lim_{n\to\infty} \sum_{i_0,\ldots,i_{n-1}=1}^{\dim\cH} \Big\langle \phi, \Big(\prod_{l=0}^{n-1} e^{-tA/n} P_{i_l}\Big) \phi \Big\rangle \Big\langle \Omega, \Big(\prod_{l=0}^{n-1} I_t\big(\tfrac{t}{n}w(i_l)\nu,\tfrac{t}{n}w(i_l)\nu\big) \Big)\Omega \Big\rangle.
		\end{align*}
		We now have for every $i_1, \hdots, i_l\in \{1, \hdots, \dim \cH\}$
		\begin{align*}
			\Big\langle \phi, \Big(\prod_{l=0}^{n-1} e^{-tA/n} P_{i_l}\Big) \phi \Big\rangle &= \sum_{i_{n}=1}^{\dim \cH} \big \langle P_{i_{n}} \phi,  \prod_{l=0}^{n-1} e^{-tA/n} P_{i_l} \phi \big \rangle \\
			&=  \sum_{i_{n}=1}^{\dim \cH}  \langle \phi, \varphi_{i_{n}} \rangle \langle \varphi_{i_{0}}, \phi \rangle  \prod_{l=0}^{n-1} e^{t v(i_l)/n} \big \langle \varphi_{i_l}, e^{tQ/n} \varphi_{i_{l+1}} \rangle \\
			&= \sum_{i_{n}=1}^{\dim \cH} \frac{1}{\dim \cH} \exp\Big(\frac{t}{n}\sum_{l=0}^{n-1} v(i_l) \Big) \cdot \prod_{l=0}^{n-1} \mathbb P(X_{(l+1)t/n} = i_{l+1}|X_{lt/n} = i_{l} ) \\
			&= \sum_{i_{n}=1}^{\dim \cH} \exp\Big(\frac{t}{n}\sum_{l=0}^{n-1} v(i_l) \Big) \cdot \mathbb P(X_0 = i_0, X_{t/n} = i_1, \hdots, X_{t} =  i_{n})
		\end{align*}
		On the other hand, an inductive application of \cref{Lemma: Product of Is} (together with the simple observation that $\langle \Omega, I_t(f, g) \Omega \rangle  = 1$ for all $f, g \in \sD(\varpi^{-1/2})$) yields
		\begin{equation*}
			\Big\langle \Omega, \Big(\prod_{l=0}^{n-1} I_t\big(\tfrac{t}{n}w(i_l)\nu,\tfrac{t}{n}w(i_l)\nu\big) \Big)\Omega \Big\rangle. = \exp\Big( \frac{t^2}{n^2}\sum_{l=0}^{n-1}\sum_{k=l+1}^{n-1} w(i_l)w(i_k) \langle \nu, e^{-t (k-l+1)\varpi/n} \nu \rangle \Big)
		\end{equation*}
		such that 
		\begin{align*}
			\langle \psi,e^{-TH}\rangle &=\lim_{n\to \infty}  \sum_{i_0, \hdots, i_{n}=1}^{\dim \cH} \mathbb P(X_0 = i_0, X_{t/n} = i_1, \hdots, X_{t} =  i_{n}) \\
			&\quad \quad \quad \quad \quad \quad \quad \quad \,\cdot \exp\Big(\frac{t}{n}\sum_{l=1}^{n-1} v(i_l) + \frac{t^2}{n^2}\sum_{l=0}^{n-1}\sum_{k=l+1}^{n-1} w(i_l)w(i_k) \big\langle \nu, e^{-t (k-l+1)\varpi/n} \nu \big\rangle \Big) \\
			&=\lim_{n\to \infty} \mathbb E\Big[ \exp\Big(\frac{t}{n}\sum_{l=0}^{n-1} v(X_{lt/n}) + \frac{t^2}{n^2}\sum_{l=1}^{n-1}\sum_{k=l+1}^{n-1} w(X_{lt/n})w(X_{kt/n}) \big\langle \nu, e^{-t (k-l+1)\varpi/n} \nu \big\rangle \Big) \Big] \\
			&= \mathbb E\Big[ \exp\Big( \int_0^t v(X_s) \mathrm ds + \int_0^t \mathrm ds \int_s^t \mathrm dr \, w(X_s) w(X_r) \big\langle \nu, e^{-(r-s)\varpi} \nu \big\rangle  \Big)\Big]. \qedhere
		\end{align*}
	\end{proof}
	
	\subsection{Proof of \cref{Theorem: Estimates Spin-Boson models}}
	We can now come to the proof of our key result on the spin boson model \cref{Theorem: Estimates Spin-Boson models}.
	
	To simplify notation, we set for measurable $x, y:[0, \infty) \to \{1, \hdots, \dim\HS\}$ and $0\leq r\leq s\leq t$
	\begin{align*}
		W_{s, t}(x) &\coloneqq \frac{1}{2}\int_{[s, t]^2} g(v-u) w(x_u)w(x_y) \, \mathrm du \mathrm dv +  \int_{[s, t]} v(x_u) \, \mathrm du \\
		W_{r, s, t}(x, y) &\coloneqq \int_{[r, s]\times [s, t]} g(v-u) w(x_u)w(y_v) \, \mathrm du \mathrm dv.
	\end{align*}
	In the following proof of \cref{Theorem: Estimates Spin-Boson models}, we will heavily exploit the fact that we start our Markov process $X$  in the uniform distribution. Since the generator $Q$ of $X$ is symmetric, the uniform distribution is the stationary distribution of $X$ and the detailed balance equations are satisfied. Hence, $X$ is time reversible i.e.\ the distribution of $(X_{T-t})_{0\leq t\leq T}$ coincides with the distribution of $(X_{t})_{0\leq t\leq T}$ for any $T\geq 0$.
	
	\begin{prop}
		\label{Proposition: Fractions of partition functions in terms of path measure}
		We have for all $0\leq s\leq t$
		\begin{equation*}
			\frac{Z_t}{Z_{s}Z_{t-s}} = \frac{\widehat{\mathbb P}_{s, t}(X_0 = Y_0)}{\mathbb P^{\otimes 2}(X_0 = Y_0)}\widehat{\mathbb E}_{s, t-s}\Big[\exp\Big(\int_0^s \mathrm du \int_{0}^{t-s} \mathrm dv \, g(u+v) w(X_u) w(Y_v) \Big)\Big|X_0 = Y_0 \Big].
		\end{equation*}
		Moreover, for all $t\geq 0$
		\begin{equation*}
			\widehat{\mathbb P}_{t, t}(X_0 = Y_0) \geq \mathbb P^{\otimes 2}(X_0 = Y_0) = 1/\dim\HS.
		\end{equation*}	
	\end{prop}
	\begin{proof}
		Throughout the proof we fix $0\le s\le t$.
		Let us define $\tilde X \coloneqq (X_{s-u})_{0\leq u \leq s}$ and  $\tilde Y \coloneqq (X_{s+v})_{0 \leq v \leq t-s}$. Notice that
		\begin{equation*}
			W_{0, s, t}(X, X) = \int_0^s \mathrm du \int_0^{t-s} \mathrm dv\, g(u+v) w(\tilde X_u)w(\tilde Y_v) \eqqcolon \widetilde W_{0,s,t}(\tilde X, \tilde Y).
		\end{equation*}	
		as well as
		\begin{equation*}
			W_{0, s}(X) = W_{0, s}(\tilde X), \quad W_{s, t}(Y) = W_{0, t-s}(\tilde Y).
		\end{equation*}
		We set $p \coloneqq  1/\dim\HS$ such that $\mathbb P(X_0 = i) = p$ for all $i \in \{1, \hdots, \dim\HS\}$. We hence have
		\begin{align*}
			Z_t &= \mathbb E\big[\operatorname{exp}\big(W_{0, s}(X) + W_{s, t}(X) + W_{0, s, t}(X, X)   \big) \big] \\
			&= \mathbb E\big[\operatorname{exp}\big(W_{0, s}(\tilde X) + W_{0, t-s}(\tilde Y) + \widetilde W_{0, s, t}(\tilde X, \tilde Y)   \big) \big] \\
			&=\sum_{i=1}^{\dim\HS} p \cdot\mathbb E \big[\operatorname{exp}\big(W_{0, s}(\tilde X) + W_{0, t-s}(\tilde Y) + \widetilde W_{0, s, t}(\tilde X, \tilde Y)  \big) \big|X_s =  i \big].
		\end{align*}
		By the Markov property and since $X$ is time reversible, we have
		\begin{equation*}
			\mathbb P\big((\tilde X, \tilde Y) \in \, \cdot  \, | X_s = i\big) = \mathbb P^{\otimes 2}\big( (X^s, Y^{t-s}) \in \, \cdot  \, | X_0 = Y_0 = i\big) 
		\end{equation*}
		where $X^s \coloneqq (X_u)_{0 \leq u \leq s}$ and $Y^{t-s} \coloneqq (Y_u)_{0 \leq u \leq t-s}$. Hence,
		we obtain
		\begin{align*}
			Z_t &=\sum_{i=1}^{\dim\HS} p \cdot\mathbb E^{\otimes 2} \big[\operatorname{exp}\big(W_{0, s}(X) + W_{0, t-s}(Y) + \widetilde W_{0, s, t}(X, Y)  \big) \big|X_0 = Y_0 = i \big] \\
			&= \frac1p \sum_{i=1}^{\dim\HS} \mathbb E^{\otimes 2} \big[\operatorname{exp}\big(W_{0, s}(X) + W_{0, t-s}(Y) + \widetilde W_{0, s, t}(X, Y)  \big) \1_{\{X_0 = Y_0 = i\}}\big] \\
			&= \frac1p Z_{s} Z_{t-s} \widehat{\mathbb E}_{s, t-s}\big[\operatorname{exp}\big(\widetilde W_{0, s, t}(X, Y)  \big) \1_{\{X_0 = Y_0\}} \big].
		\end{align*}
		which readily implies the first equality.

		To prove the claimed inequality, it is sufficient to notice that
		\begin{align*}
			\widehat{\mathbb P}_{t, t}(X_0 = Y_0) &=  \frac{1}{Z_t^2} \sum_{i=1}^{\dim\HS} \mathbb E^{\otimes 2}\big[\exp(W_{0, t}(X) + W_{0, t}(Y)) \1_{\{X_0 = Y_0 = i\}} \big] \\
			&= \frac{1}{Z_t^2} \sum_{i=1}^{\dim\HS} \mathbb E\big[\exp(W_{0, t}(X)) \1_{\{X_0 = i\}} \big]^2 \\
			&\geq \frac{1}{Z_t^2} \frac{1}{\dim\HS} \Big(\sum_{i=1}^{\dim\HS} \mathbb E\big[\exp(W_{0, t}(X)) \1_{\{X_0 = i\}} \big] \Big)^2 = \frac{1}{\dim\HS}
		\end{align*}
		by the usual inequality between the $1$- and the $2$-norm on $\R^d$.
	\end{proof}
	We can now give the
	\begin{proof}[Proof of \cref{Theorem: Estimates Spin-Boson models}]
		We first prove \cref{eq:SBlowerbound}. By applying \cref{Proposition: Fractions of partition functions in terms of path measure} and Jensens inequality, we find
		\begin{align*}
			\frac{Z_{2t}}{Z_t^2} &\geq \widehat{\mathbb E}_{t, t}\Big[\exp\Big(\int_0^t \mathrm du \int_{0}^{t} \mathrm dv \, g(u+v) w(X_u) w(Y_v) \Big)\Big|X_0 = Y_0 \Big] \\
			&\geq \int_0^t \mathrm du \int_{0}^{t} \mathrm dv \, g(u+v) \widehat{\mathbb E}_{t, t}[w(X_u) w(Y_v)|X_0 = Y_0 = 0].
		\end{align*}
		After taking the logarithm, \cref{eq:SBlowerbound} follows with \cref{Theorem: Convergence of averaged partition function}.
		
		Let us now again set $p\coloneqq 1/\dim\HS$.
		We have for all $0\leq s\leq t$
		\begin{align*}
			\frac{Z_{s} Z_{t-s}}{Z_t} &= \frac{1}{Z_t} \sum_{i, j=1}^{\dim\HS} \mathbb E[\exp(W_{0, s}(X))|X_0 = i]\cdot \mathbb E[\exp(W_{0, t-s}(X))|X_0 = j]\cdot p^2 \\
			&\geq \frac{1}{Z_t} \sum_{i=1}^{\dim\HS} \mathbb E[\exp(W_{0, s}(X))|X_0 = i]\cdot \mathbb E[\exp(W_{0, t-s}(X))|X_0 = i]\cdot p^2 \\
			&= \frac{1}{Z_t} \sum_{i=1}^{\dim\HS} \mathbb E[\exp(W_{0, s}(X))|X_s = i]\cdot \mathbb E[\exp(W_{s, t}(X))|X_s = i]\cdot p^2
		\end{align*}
		where we used in the last equality that $X$ is time reversible. With the Markov property, we hence obtain
		\begin{align*}
			\frac{Z_{s} Z_{t-s}}{Z_t} &\geq \frac{1}{Z_t} \sum_{i=1}^{\dim\HS} \mathbb E[\exp(W_{0, s}(X))\exp(W_{s, t}(X))|X_s = i]\cdot p^2 \\
			&=\sum_{i=1}^{\dim\HS}\widehat{\mathbb E}_t[\exp(-W_{0, s, t}(X, X))\1_{\{X_s = i\}}]\cdot p \\
			&= p \cdot \widehat{\mathbb E}_t[\exp(-W_{0, s, t}(X, X))].
		\end{align*}
		Notice that
		\begin{align*}
			\int_0^t \mathrm ds \, W_{0, s, t}(X, X) &= \int_0^t \mathrm ds \int_{[0, t]^2} \mathrm du \mathrm dv \, \1_{\{0<u<s<v<t\}} g(v-u)w(X_u) w(X_v) \\
			&=\frac{1}{2}\int_{[0, t]^2}\mathrm du \mathrm dv \, |v-u| g(v-u)w(X_u) w(X_v).
		\end{align*}
		By applying Jensens inequality, we hence obtain
		\begin{align*}
			\frac{1}{t} \int_0^t \frac{Z_{s} Z_{t-s}}{Z_t} \, \mathrm ds &\geq p\exp\Big(-\frac{1}{t}\int_0^t \widehat{\mathbb E}_t[W_{0, s, t}(X, X))] \mathrm ds  \Big) \\
			&= p\operatorname{exp}\Big(-\frac{1}{2t}\int_{[0, t]^2} |u-v|g(u-v) \widehat{\mathbb E}_{t}[w(X_u) w(X_v)] \, \mathrm du \mathrm dv\Big). 
		\end{align*}
		After taking the logarithm, \cref{Theorem: Convergence of averaged partition function} finally yields \cref{eq:SBupperbound}.
	\end{proof}
	
	\subsection{Proof of \cref{Theorem: estimate vacuum overlap generalized Spin-Boson,Corollary: Existence criterion generalized Spin-Boson}}
	To derive \cref{Theorem: estimate vacuum overlap generalized Spin-Boson} from \cref{Theorem: Estimates Spin-Boson models},
	we rewrite the right hand side in the latter in terms of the boson number operator $\mathbf N$.
	
	Recalling that $H_\eps$  is obtained by replacing $\varpi$ in \cref{Equation: Definition of generalized Spin-Boson Hamiltonian} by $\varpi+\eps$, the Feynman--Kac formula \cref{Proposition: Feynman-Kac} yields
	\begin{equation}
		\label{Equation: Feynman-Kac for regularized Hamiltonian}
		\big\langle \psi, e^{-TH_\varepsilon}\psi \big\rangle = \mathbb E\Big[\exp\Big( \frac{1}{2}\int_{[0, T]^2}g_\varepsilon(t-s) w(X_s)w(X_t) \, \mathrm ds \mathrm dt +  \int_{[0, T]} v(X_s) \, \mathrm ds \Big) \Big] = Z_{\varepsilon, T}
	\end{equation}
	where $Z_{\varepsilon, T}$ is the normalization constant of the measure $\widehat{\mathbb P}_{\varepsilon, T}$ obtained by replacing $g$ with
	\begin{equation*}
		g_\varepsilon(t) = e^{-\varepsilon|t|}g(t).
	\end{equation*}
	For $\varepsilon\ge0$, we denote $E_\varepsilon \coloneqq \inf \sigma(H_\varepsilon)$ and $\rho_\varepsilon \coloneqq \langle \psi, \1_{\{E_\varepsilon\}}(H_\varepsilon) \psi \rangle$.
	
	\begin{prop}
		\label{Proposition: Representation of upper bound as expected boson number}
		Assume that $\inf\sigma(\varpi) > 0$ and let $\phi$ denote the ground state of $H$,
		which exists by \cref{Corollary: Existence in infra-red regular case}.
		Then
		\begin{equation*}
			\lim_{T \to \infty} \frac{1}{2T}\int_{[0, T]^2} |t-s|g(t-s) \widehat{\mathbb E}_{T}[w(X_s) w(X_t)] \, \mathrm ds \mathrm dt = \big \langle \phi, \mathbf N \phi \big \rangle.
		\end{equation*}
	\end{prop}
	
	\begin{proof}
		As $\sD(H) \subset \sD(\mathbf N)$, the operator $\mathbf N$ is, as a consequence of the closed graph theorem, relatively bounded with respect to $H$,
		see for example \cite{Schmudgen.2012} for the simple proof.
		In the following, let $C>0$ be such that
		\begin{equation}
			\label{Equation: relative bound}
			\|\mathbf N\phi \|^2 \leq C \|\phi\|^2 +  	C\|H \phi \|^2
		\end{equation}
		for all $\phi \in \mathcal D(H)$. Since $\partial_\varepsilon g_\varepsilon(t)|_{\varepsilon = 0} = -|t|g(t)$, we obtain
		\begin{equation*}
			- \frac{1}{T} \partial_\varepsilon \log \,\big\langle \psi, e^{-TH_\varepsilon}\psi \big\rangle \big|_{\varepsilon = 0} = \frac{1}{2T}\int_{[0, T]^2} |t-s|g(t-s) \widehat{\mathbb E}_{T}[w(X_s) w(X_t)] \, \mathrm ds \mathrm dt.
		\end{equation*}
		On the other hand, by Duhamels formula
		\begin{equation*}
			- \frac{1}{T} \partial_\varepsilon \log \,\big\langle \psi, e^{-TH_\varepsilon}\psi \big\rangle\big|_{\varepsilon = 0} = \frac{1}{T}\int_0^T \frac{\big\langle \psi, e^{-sH} \mathbf N e^{-(T-s)H} \psi \big\rangle}{\big\langle \psi, e^{-TH}\psi \big\rangle } \, \mathrm ds.
		\end{equation*}
		Write $\tilde H \coloneqq H - E$. By multiplying denominator and enumerator by $e^{TE} = e^{(T-s)E} \cdot e^{sE}$, we obtain
		\begin{equation*}
			\frac{1}{2T}\int_{[0, T]^2} |t-s|g(t-s) \widehat{\mathbb E}_{T}[w(X_s) w(X_t)] \, \mathrm ds \mathrm dt = \frac{1}{T}\int_0^T \frac{\big\langle \psi, e^{-s \tilde H} \mathbf N e^{-(T-s) \tilde H} \psi \big\rangle}{\big\langle \psi, e^{-T \tilde H}\psi \big\rangle } \, \mathrm ds.
		\end{equation*}
		Since  
		\begin{equation*}
			\lim_{T\to \infty} \big\langle \psi, e^{-T \tilde H}\psi \big\rangle = \lim_{T\to \infty} \int_{[E, \infty)} e^{-T(x-E)} \, \mu(\mathrm dx) = \mu(\{E\}) = \|P_E \psi \|^2
		\end{equation*}
		it is sufficient to show that
		\begin{equation}
			\label{Equation: some equation}
			\lim_{T\to \infty} \frac{1}{T}\int_0^T \big\langle \psi, e^{-s \tilde H} \mathbf N e^{-(T-s) \tilde H} \psi \big\rangle \, \mathrm ds =  \|P_E \psi \|^2 \big \langle \phi, \mathbf N \phi \big \rangle
		\end{equation}
		in order to conclude the proof. Let us write
		\begin{equation*}
			\psi = \psi_1 + \psi_2 \quad \text{ where } \quad  \psi_1 = P_E \psi \quad \text{ and } \quad  \psi_2 = (1-P_E)\psi.
		\end{equation*}
		Since $\psi_1$ is an eigenvector of $\tilde H$ to the eigenvalue $0$, we have
		\begin{align*}
			\big\langle \psi, e^{-s \tilde H} \mathbf N e^{-(T-s) \tilde H} \psi \big\rangle &= \big\langle  e^{-s \tilde H}  \psi_1, \mathbf N e^{-(T-s) \tilde H} \psi_1 \big\rangle + R(s, T) \\
			&= \big\langle \psi_1, \mathbf N \psi_1 \big\rangle + R(s, T) \\
			&= \|P_E \psi \|^2 \big \langle \phi, \mathbf N \phi \big \rangle + R(s, T)
		\end{align*}
		where
		\begin{equation*}
			R(s, T) \coloneqq \sum_{(i, j)\neq (1, 1)}  \big\langle  e^{-s \tilde H}  \psi_i, \mathbf N e^{-(T-s) \tilde H} \psi_j\big \rangle.
		\end{equation*}
		Let $\nu_1, \nu_2$ denote the spectral measures of $\psi_1$ and $\psi_2$ with respect to $\tilde H$. Using \cref{Equation: relative bound}, we hence have for all $i, j\in \{1, 2\}$
		\begin{align*}
			\big|\big\langle  e^{-s \tilde H}  \psi_i, \mathbf N e^{-(T-s) \tilde H} \psi_j\big \rangle \big|^2 \leq 	C \big\| e^{-s \tilde H}  \psi_i \big\|^2 \Big(\big\|e^{-(T-s) \tilde H} \psi_j \big\|^2 + \big\|H e^{-(T-s) \tilde H} \psi_j \big\|^2\Big) \\
			= C \Big(\int_{[0, \infty)} e^{-2sx} \nu_i(\mathrm dx) \Big) \Big(\int_{[0, \infty)} \big(1 + (x+ E)^2\big) e^{-2(T-s)x} \nu_j(\mathrm dx)  \Big)
		\end{align*}
		Notice that $\int_{[0, \infty)} x^2 \, \nu_j(\mathrm dx) < \infty$ since $\psi_j  \in \mathcal D(H)$.
		Let $\varepsilon>0$. Since $\nu_2(\{0\}) = 0$, the above implies that there exists some constant $C>0$ such that for all sufficiently large $T>0$
		\begin{equation*}
			|R(s, T)| \leq 
			\begin{cases}
				\varepsilon \quad  &\text{ if }s\geq \sqrt{T} \text{ and } T-s\geq \sqrt{T} \\
				C \quad &\text{ else}
			\end{cases}
		\end{equation*}
		Hence,
		\begin{equation*}
			\limsup_{T\to \infty} \frac{1}{T}\int_0^T |R(s, T)| \, \mathrm ds \leq \varepsilon + \limsup_{T\to \infty} \frac{2 C \sqrt{T}}{T} = \varepsilon.
		\end{equation*}
		Since $\varepsilon>0$ was arbitrary, we hence obtain
		\begin{equation*}
			\lim_{T\to \infty} \frac{1}{T}\int_0^T R(s, T) \, \mathrm ds = 0
		\end{equation*}
		and \cref{Equation: some equation} follows.
	\end{proof}
	To conclude the proof of \cref{Corollary: Existence criterion generalized Spin-Boson}, 
	we employ the following two simple observations on the limit $\eps\downarrow 0$.
	Whereas they can also be proven using functional analytic techniques,
	we emphasize that we here exclusively employ the Feynman--Kac formula.
	\begin{lem}
		\label{Lemma: continuity ground-state energy}
		We have $\lim_{\varepsilon \downarrow 0}E_\varepsilon = E$.
	\end{lem}
	\begin{proof}
		Let $\mu_\varepsilon$ denote the spectral measure of $H_\varepsilon$ with respect to $\psi$.
		For every $T\geq 0$ we have with the Feynman-Kac formula \cref{Equation: Feynman-Kac for regularized Hamiltonian} that $Z_{T, \varepsilon} \to Z_{T}$ as $\varepsilon\downarrow 0$ and hence $\mu_\varepsilon \to \mu$ weakly as $\varepsilon\downarrow 0$. Let $\delta>0$ and $f \in C_c(E-\delta, E+\delta)$ with $f(E) > 0$. Then $\mu(f)>0$ and hence $\mu_\varepsilon(f) > 0$ for all sufficiently small $\varepsilon>0$ i.e.\ $E_\varepsilon \leq E + \delta$ for all sufficiently small $\varepsilon>0$. Since $\mathbf N$ is positive definite, we have $E_\varepsilon \geq E$ for all $\varepsilon>0$. We hence obtain $E_\varepsilon \to E$ as $\varepsilon \downarrow 0$.
	\end{proof}
	\begin{lem}
		\label{Lemma: continuity vacuum overlap}
		We have $\rho \geq \limsup_{\varepsilon \downarrow 0}\rho_\varepsilon$.
	\end{lem}
	\begin{proof}
		Let $\nu_{\varepsilon}$ denote the spectral measure of $\tilde H_\varepsilon = H_\varepsilon - E_\varepsilon$ with respect to $\psi$. The Feynman-Kac formula \cref{Equation: Feynman-Kac for regularized Hamiltonian,Lemma: continuity ground-state energy} imply that
		\begin{equation*}
			\int_\mathbb R e^{-tx} \nu_\varepsilon(\mathrm dx) = e^{tE_\varepsilon} \langle \psi, e^{-tH_\varepsilon} \psi \rangle \to e^{tE} \langle \psi, e^{-tH} \psi \rangle = \int_\mathbb R e^{-tx} \nu_0(\mathrm dx)
		\end{equation*}		
		as $\varepsilon \downarrow 0$ and hence $\nu_\varepsilon \to \nu_0$ weakly as $\varepsilon \downarrow 0$. Hence, by the Portmanteu Theorem
		\begin{equation*}
			\limsup_{\varepsilon\downarrow0} \rho_\varepsilon = \limsup_{\varepsilon\downarrow0}\mu_\varepsilon(\{E_\varepsilon\}) = \limsup_{\varepsilon\downarrow 0} \nu_\varepsilon(\{0\}) \leq \nu_0(\{0\}) = \mu(\{E\}) = \rho. \qedhere
		\end{equation*}
	\end{proof}
	
	\subsection{Proof of \cref{Proposition: Estimates in sepcial case of Spin-Boson}}
	
	We finish our discussion by giving the proof of \cref{Proposition: Estimates in sepcial case of Spin-Boson} in which we review and strengthen the upper estimates on $\log(1/\rho)$ given in \cref{Theorem: estimate vacuum overlap generalized Spin-Boson,Theorem: Estimates Spin-Boson models} for the special case of the SSB model.
	
	\begin{proof}[Proof of \cref{Proposition: Estimates in sepcial case of Spin-Boson}]
		For the SSB model, the process $X$ is a continuous time random walk on $\{-1, 1\}$ with $\operatorname{Exp(1)}$ distributed waiting times. Let $\xi$ be the point process of jumping times of $X$. Then $\xi$ is a Poisson point process whose intensity measure is the Lebesgue measure. Notice that for any $0\leq s\leq t$
		\begin{equation}
		\label{Equation: Spin product in terms of jumping process}
			X_s X_t = (-1)^{\xi((s, t])}.
		\end{equation}
		Hence, the for any $0\leq s \leq t$ the processes $(X_{u}X_{v} )_{u, v\in [0, s]}$ and $(X_{u}X_{v})_{u, v\in [s, t]}$ are independent and we have by translation invariance of $\xi$
		\begin{equation*}
			Z_s Z_{t-s} = e^{-t}\mathbb E\Big[\exp\Big(\frac{1}{2}\int_{[0, s]^2}g(v-u) X_u X_v\, \mathrm du \mathrm dv + \frac{1}{2}\int_{[s, t]^2}g(v-u) X_u X_v\, \mathrm du \mathrm dv\Big)  \Big].
		\end{equation*}
		Hence, we have
		\begin{equation*}
			\frac{Z_s Z_{t-s}}{Z_t} = \widehat{\mathbb E}_t\Big[\exp\Big(-\int_{[0, t]^2}\1_{\{0\leq u \leq s < v \leq t\}}g(v-u) X_u X_v\, \mathrm du \mathrm dv \Big)\Big].
		\end{equation*}
		As in the proof of \cref{Theorem: Estimates Spin-Boson models}, the estimate \cref{Equation: Existence criterion Spin-Boson} follows by Markov's inequality and \cref{Theorem: Convergence of averaged partition function}. The estimate \cref{Equation: Existence criterion Spin-Boson in terms of number operator} directly follows from \cref{Equation: Existence criterion Spin-Boson,Proposition: Representation of upper bound as expected boson number}.
	\end{proof}

	
	\bibliographystyle{halpha-abbrv}
	\bibliography{../../Literature/00lit}
	
\end{document}